\documentclass[reqno]{amsart}

\usepackage{comment,amssymb,latexsym}
\usepackage{todd-defs}

\newcounter{mnotecount}[section]

\newcommand{\dt}[1]{\frac{d^{#1}\:}{dt^{#1}}\Bigl|_{t=0} }

\newcommand{\AND}{{\quad\text{and}\quad}}
\newcommand{\where}{\quad\text{where}\quad}

\newcommand{\Half}{\ensuremath{\textstyle\frac{1}{2}}}

\newcommand{\norm}[1]{\|#1\|}

\newcommand{\Abb}{\mathbb{A}}
\newcommand{\Ebb}{\mathbb{E}}
\newcommand{\Pbb}{\mathbb{P}}
\newcommand{\Rbb}{\mathbb{R}}
\newcommand{\Mbb}{\mathbb{M}}
\newcommand{\Zbb}{\mathbb{Z}}

\newcommand{\Ec}{\mathcal{E}}
\newcommand{\Bc}{\mathcal{B}}
\newcommand{\Cc}{\mathcal{C}}
\newcommand{\Dc}{\mathcal{D}}
\newcommand{\Fc}{\mathcal{F}}
\newcommand{\Gc}{\mathcal{G}}
\newcommand{\Lc}{\mathcal{L}}
\newcommand{\Mcal}{\mathcal{M}}
\newcommand{\Nc}{\mathcal{N}}

\newcommand{\Uc}{\mathcal{U}}
\newcommand{\Wc}{\mathcal{W}}
\newcommand{\Oc}{\mathcal{O}}
\newcommand{\Oct}{\widetilde{\mathcal{O}}}
\newcommand{\Yc}{\mathcal{Y}}
\newcommand{\Zc}{\mathcal{Z}}

\newcommand{\Ft}{\tilde{F}}
\newcommand{\Gt}{\tilde{G}}

\newcommand{\Lb}{\bar{L}}
\newcommand{\taub}{\bar{\tau}}

\newcommand{\Ub}{\bar{U}}
\newcommand{\Vb}{\bar{V}}

\newcommand{\Gv}{\mathbf{G}}

\newcommand{\ep}{\lambda}
\newcommand{\phit}{\tilde{\phi}}
\newcommand{\phiv}{{\boldsymbol \phi}}
\newcommand{\Phiv}{{\boldsymbol \Phi}}
\newcommand{\psiv}{{\boldsymbol \psi}}

\newcommand{\EdB}{E_{\partial}}
\newcommand{\AdB}{A_{\partial}}

\newcommand{\id}{\mathbf{id}}
\newcommand{\del}{\partial}

\renewcommand{\Re}{{\mathbb R}}         
\newcommand{\la}{\langle}               
\newcommand{\ra}{\rangle}               
\newcommand{\half}{\frac{1}{2}}         

\newcommand{\Lie}{\mathcal L}           

\newcommand{\tr}{\text{\rm tr}}         

\newcommand{\Bo}{\mathcal B}
\newcommand{\Sp}{\mathcal S}

\newcommand{\LL}{\mathcal L}

\newcommand{\EE}{\mathcal E}

\newcommand{\normal}{\nu}

\renewcommand{\div}{\text{div}}

\newcommand{\Om}{\Omega}
\newcommand{\dO}{\partial \Omega}

\newcommand{\ddnu}{\frac{\partial}{\partial \normal}}
\newcommand{\ddnuy}{\frac{\partial}{\partial \normal_Y}}

\renewcommand{\SS}{\mathcal S}
\newcommand{\DD}{\mathcal D}
\newcommand{\VV}{\mathcal V}

\newcommand{\tg}{\widetilde{g}}
\newcommand{\ta}{\widetilde{a}}
\newcommand{\tf}{\widetilde{f}}
\newcommand{\tw}{\widetilde{w}}
\renewcommand{\th}{\widetilde{h}}

\newcommand{\hf}{\widehat{f}}
\newcommand{\hw}{\widehat{w}}

\theoremstyle{plain}
\newtheorem{thm}{Theorem}[section]
\newtheorem{lem}[thm]{Lemma}
\newtheorem{prop}[thm]{Proposition}
\newtheorem{cor}[thm]{Corollary}

\theoremstyle{definition}
\newtheorem{Def}[thm]{Definition}

\theoremstyle{remark}
\newtheorem{rem}[thm]{Remark}
\newtheorem{ex}[thm]{Example}

\title[Dynamical elastic bodies in Newtonian gravity]{Dynamical elastic
  bodies in Newtonian gravity}

\author[L. Andersson]{Lars Andersson}
\email{lars.andersson@aei.mpg.de}
\address{Albert Einstein Institute, Am M\"uhlenberg 1, D-14476 Potsdam,
  Germany}
\author[T.A. Oliynyk]{Todd A. Oliynyk${}^\dagger$}
\thanks{${}^\dagger$ Partially supported by the ARC grant DP1094582
and an MRA grant.}
\address{School of Mathematical Sciences\\
Monash University, VIC 3800\\
Australia}
\email{todd.oliynyk@sci.monash.edu.au}

\author[B. Schmidt]{Bernd G. Schmidt} \email{bernd@aei.mpg.de}
\address{Albert Einstein Institute,
Am M\"uhlenberg 1, D-14476 Potsdam, Germany}

\begin{document}


\begin{abstract}
Well-posedness for the initial value problem for a self-gravi\-ta\-ting
elastic body with free boundary in Newtonian gravity is proved. In the
material frame, the Euler-Lagrange  equation
becomes, assuming suitable constitutive
properties for the elastic material, a fully non-linear elliptic-hyperbolic
system
with boundary conditions of Neumann type. For systems of this type, the
initial data must satisfy compatibility conditions in order to achieve
regular solutions. Given a relaxed reference configuration and a sufficiently
small Newton's constant, a neigborhood of initial data satisfying the
compatibility conditions is constructed.
\end{abstract}

\maketitle

\section{ Introduction} \label{sec:intro}
In Newtonian physics, the two-body problem is solvable for point particles
moving around their common center of gravity on Kepler ellipses. However, if
one considers
extended bodies, the situations changes
drastically. Assuming that a solution exists for extended bodies, one can
show that the centers of mass of the bodies move as point particles, but
existence  could not be established for a long time.

The first sucessful attack on the problem was made by Leon Lichtenstein
\cite{lichtenstein}
who considered self-gravitating
fluid bodies moving on circles about their center of gravity. In this case
the Euler equations for a self-gravitating system become time independent in a
coordinate system co-moving with the bodies. For the case of small
bodies or widely separated large bodies, Lichtensten showed the existence
of such solutions.

Well-posednedness
for the Cauchy problem for fluid bodies with free boundary was proved only
recently. Lindblad  \cite{lindblad:comp}
proved well-posedness for a non-relativistic
compressible liquid body (i.e. positive boundary
density) with free boundary. In this paper one can also find references to earlier work. A different proof
of the result of Lindblad that is valid for both the relativistic and non-relativistic cases was given by Trakhinin \cite{trakhinin:2009}.

The more singular case of fluids with vanishing boundary density is discussed
in \cite{CK3D,JM3D}.
Unfortunately, the problem of proving well-posedness for self-gravitating
compressible fluid bodies with free boundary is still open in general relativity and
even in Newtonian gravity. However, see \cite{lindblad:nordgren:2009} for
the case of self-gravitating incompressible fluids.
See also
\cite{rendall:1992, ChoquetBruhat:2006qn, kind:ehlers} for
results dealing with various restricted versions of the Cauchy problem for
fluid bodies in general relativity.

One can argue that the Cauchy problem should be simpler to handle
if the bodies consist of
elastic material, since such a body can, at least in the case of small
bodies,
be said to 'have a shape of its own'.
In this paper we solve the boundary initial value problem for
self-gravitating
deformations
of a relaxed body. In particular, consider a relaxed (i.e. with
vanishing stress)
elastic body in the absence of gravity.
Then we show existence of  a self-gravitating
solution for initial data close to those of the static relaxed body and for
a small gravitational
constant. Thus, the motion of the body will consist of small nonlinear
oscillations around
the equilibrium
solution.

The Cauchy problem for a Newtonian elastic body with free boundary is, under
suitable constitutive assumptions on the elastic material, a fully
non-linear elliptic-hyperbolic initial-boundary value problem with boundary
values of Neumann type. The case of purely hyperbolic problems of this type
is covered by a theorem of Koch \cite{koch}.
The method of proof of this theorem can be adapted to
include the non-local elliptic
terms in the which appear in the system of differential equations considered
in this paper via the Newtonian
gravitational field.

In order to achieve a regular evolution for the initial-boundary value
problem, the initial data must satisfy certain compatibility conditions. The
problem of constructing an open neighborhood of initial data satisfying
the compatibility conditions given one such data set is discussed in the work of Koch.
The work of Lindblad contains a related
discussion for the fluid case. For the case of a Newtonian elastic body,
we construct, for small values of Newton's constant,
initial data satisfying suitable compatibility conditions
near data for a relaxed elastic body in the absence of gravity.
It should be noted that due to
the presence of the Newtonian gravitational potential this problem is
non-local.

The situation considered in this paper has several
interesting generalizations. If we consider a
large static, self-gravitating body, a relaxed state may not
exist. Restricting to the spherically symmetric case, solutions to
the field equations can be found
by ODE techniques \cite{park:2000}. If we perturb the data slightly,
there should exist solutions with small oscillations around the equilibrium
configuration. Our approach allows one to consider also such more complicated
problems; the essential difficulty is the construction of the initial data
satisfying the compatibility conditions and not the time evolution. However,
the results in this paper do no immediately cover these more general
situations. In particular, a proof of a suitable version of the well-posedness
result for the Cauchy problem, in these situations, requires a
more general treatment of potential theory in Riemannian manifolds. This
problem will be considered in a separate paper.

Given a solution $u$
to the Cauchy problem for the Newtonian elastic body which exists for for
times $t \in [0,T]$, one may conclude from
the main result of this paper, see theorem \ref{lethm}, that by choosing
initial data sufficiently close to that of $u$, the time of existence
of the resulting solution can be arbitrarily close to $T$.
In particular, given a
a static solution, there exist nearby data such that the corresponding
solution to the Cauchy problem has a time of existence not less than any
given time $T$.
The proof of
existence of static self-gravitating bodies in Newtonian theory
\cite{BS:PRS2003} could be used together with
a generalization of theorem \ref{lethm} that allows for more general initial
data, as alluded to above, as well as a generalization of the construction
of initial data to show that
shows that nearby data define solutions which exist up to
some time $T$.

The treatment of two fluid balls in steady rotation due to Lichtenstein,
which was mentioned above, has been
generalized to the case of elastic bodies, see \cite{beig:schmidt:celest}.
Using these solutions in the
manner just described, it is possible to obtain
classes of solutions of the two body problem
in any prescribed finite time interval.
We leave these problems for later investigations.

Global existence for small data is an important question. For unbounded fluids in 3-dimensions,
it is now known that shocks form for a large class of initial data \cite{sideris:shocks,christodoulou:shocks},
and it is widely believed that for most equations of state shocks will always form for arbitrary small perturbation
of constant density state. On the other hand, there do exists
small data global existence results for
unbounded elastic bodies provided the material satisfies
certain additional conditions, cf.  \cite{sideris:nullcond:1996,domanski:ogden:2006}.
For elastic bodies of finite extent, it appears to be an open question
whether there exists solutions that are global in time or under what conditions
shocks form.

\subsection*{Overview of this paper}
The paper is organized as follows.
Section \ref{newelas} describes Newtonian elasticity, sets up the basic
equations and gives the conditions we impose on the material. We derive the
equations in the material frame
and in spacetime from a variational
principle.
Section \ref{compatsec}
deals with the problem of finding solutions to the
compatibility conditions needed for the proof of local existence.
Due to the non-local terms in the
equations, these condition imply conditions on the Cauchy data on the whole
initial surface.  To find initial data satisfying these conditions, we make
use of some results from potential theory. These are developed in section \ref{sec:potential}.
Further, the Poisson equation must
be studied in the material frame, see section \ref{newt}.
The results concerning the linearized elasticity operator which are needed
can be found in section \ref{leop}.
Section 4 generalizes Koch's theorem and proves our main theorem.
Appendix \ref{sec:prelim} contains some background material for
the function spaces used in this paper.

\sect{newelas}{Newtonian elasticity}

\subsection{Kinematics} \label{sec:kinematics}
The \emph{body} $\Bo$ is an open, connected, and bounded
set with a $C^\infty$ boundary in Euclidean space $\Re^3_\Bo$. We refer to
$\Re^3_\Bo$ as the \emph{extended body}.
We consider \emph{configurations}, i.e. maps
$f: \Re^3_{\Sp} \to \Bo$, and deformations
$\phi: \Bo \to
\Re^3_{\Sp}$ with
\begin{equation}\label{eq:fphirel}
f \circ \phi = \id_{\Bo} .
\end{equation}
Thus, the physical body is the domain in \emph{space} $\Re^3_\Sp$ given by
$f^{-1}(\Bo) = \phi(\Bo)$.

We shall make use of the extension $\phit$ of $\phi$ to a map $\Re^3_{\Bo}
\to \Re^3_{\Sp}$ and let $(X^A)_{A=1,2,3}$ and $(x^i)_{i=1,2,3}$ denote global Cartesian coordinates
on the body $\Re^3_{\Bo}$ and the configuration $\Re_{\Sp}$ spaces, respectively.
Since we shall consider the
Newtonian dynamics of a body, we let $f, \phi$ depend on time, denoted by
$t$. Equation
(\ref{eq:fphirel}) gives
\begin{equation}\label{eq:fphirel-time}
f^A (t, \phi(t,X)) = X^A \quad \text{ in } \Bo \AND \phi^i (t,f(t,x)) =
x^i \quad \text{ in } f^{-1} (\Bo) .
\end{equation}

Writing  $x^\mu = (t, x^i)$, we introduce $f^A{}_\mu = \partial_\mu f^A$ and $\phi^k{}_A = \partial_A
\phi^k$. In particular, $f^A{}_0 = \partial_t f^A$.
We have
\begin{equation}\label{eq:phif-inv}
\phi^k{}_A f^A{}_\ell = \delta^k{}_\ell \quad  \text{ and } \quad f^B{}_k \phi^k{}_A =
\delta^B{}_A
\end{equation}
where these expressions are defined.
This implies
\begin{equation}\label{dphidf}
\frac{\partial \phi^i{}_A}{\partial f^B{}_k} = - \phi^i{}_B \phi^k{}_A,
\quad \text{ and  }
\quad \frac{\partial f^B{}_k}{\partial \phi^i{}_A} = - f^B{}_i f^A{}_k .
\end{equation}

Let
\begin{equation}\label{eq:chidef}
\chi_{f^{-1}(\Bo)} = \left\{ \begin{array}{ll} 1 & \text{in } f^{-1}(\Bo) \\
0 & \text{in } \Re^3_{\Sp} \setminus f^{-1}(\Bo) \end{array} \right .
\end{equation}
be the indicator function of the support of the
physical body. Using the above identities, one may calculate the variation
  of $\chi_{f^{-1}(\Bo)}$ with respect to $f^A$,
\begin{equation}\label{eq:ddchi}
\frac{\partial \chi_{f^{-1}(\Bo)}}{\partial (f^A)} = \phi^i{}_A \partial_i
\chi_{f^{-1}(\Bo)}
\end{equation}

Let $H^{AB} = f^A{}_{,i} f^B{}_{,j} \delta^{ij}$ and define
$H_{AB}$ by $H_{AB} H^{BC} = \delta_A{}^C$. Then $H_{AB} = \phi^i{}_A
\phi^j{}_B \delta_{ij}$. We have
\begin{equation}\label{eq:dHdf}
\frac{\partial H^{AB}}{\partial (f^C{}_k)} = 2 \delta^{kn} \delta^{(A}{}_C
f^{B)}{}_n .
\end{equation}
Differentiating (\ref{eq:phif-inv}) gives the usual formulas for the
derivative of the inverse,
\begin{align}\label{eq:dmuphi}
\partial_\mu \phi^i{}_A &= - \phi^i{}_B \partial_\mu f^B{}_k \phi^k{}_A \,,\\
\partial_B f^A_i &= - f^A{}_j \partial_B \phi^j{}_C f^C{}_i \label{eq:dBf} \,.
\end{align}
We let
$v^\mu \partial_\mu = \partial_t + v^i \partial_i$
be defined by $v^\mu f^A{}_\mu = 0$. This determines  the vector field $v^i(x)$
on $\Re^3_{\Sp}$ uniquely in terms of $f$.
The velocity field
$v^\mu \partial_\mu$ describes the trajectories of
material particles. From the relation $v^\mu \partial_\mu f^A = 0$, we get
$$
v^i = - \phi^i{}_A f^A{}_0 .
$$
On the other hand, time differentiating (\ref{eq:fphirel-time}) gives
$$
v^i(t,x) = (\partial_t \phi^i)(t,f(t,x)) .
$$
Thus, we have
\begin{align*}
\partial_t^2 \phi^k(t,X) &= \partial_t (v^k(t,\phi(t,X)))  \\
&= (\partial_t v^k)(t,\phi(t,X)) + \partial_\ell v^k (t,\phi(t,X)) \partial_t
\phi^\ell(t,X) \\
&= (v^\mu \partial_\mu v^k)(t,\phi(t,X)) .
\end{align*}
Further,
\begin{equation}\label{eq:vmunufA}
v^\mu v^\nu \partial_\mu \partial_\nu f^A = - v^\mu \partial_\mu v^k f^A{}_k,
\end{equation}
which shows that
\begin{equation}\label{eq:HvmunufA}
H_{AB} v^\mu v^\nu \partial_\mu \partial_\nu f^B = - v^\mu \partial_\mu v^m
\delta_{mn} \phi^n{}_A
.
\end{equation}
The body  $\Bo$ carries a reference
volume element defined by a 3-form $V_{ABC}$ on $\Bo$.
The number density $n$ is defined, cf.
\cite[eq. (3.2)]{ABS}, by
$$
f^A{}_{i} f^B{}_{j} f^C{}_{k} V_{ABC}(f(x)) = n(x) \epsilon_{ijk}(x)
$$
where $\epsilon_{ijk}$ is the volume element of the Euclidean metric
$\delta_{ij}$ on $\Re^3_{\Sp}$. Since we are
considering the Newtonian case with Euclidean geometry, we have simply
$$
n = \det Df
.
$$
The mass density is
\begin{equation}\label{eq:massdensity}
\rho = n m
\end{equation}
where $m$ is the specific mass of the particles, i.e. $m V$ is the mass
density\footnote{\label{foot:1}Our techniques allow for $m$ to be a function on the body, $m=m(X)$, but
for simplicity, we will assume that $m$ is constant.}
for the material in its natural state,
see \cite{ABS}. See also \cite{BeSch03}
where the specific mass is denoted $\rho_0$. We have the relation
\begin{equation}\label{eq:dndf}
\frac{\partial n}{\partial f^A{}_i} = n \phi^i{}_A
.
\end{equation}

The following equations for $n$ follow from the above definitions, using
(\ref{eq:dndf}),
\begin{equation} \label{eq:continuity}
\partial_\mu (n v^\mu) = 0
,
\end{equation}
and
\begin{equation} \label{eq:nphi}
\partial_k (n \phi^k{}_A) = 0
.
\end{equation}
The elastic material is
described by the stored energy function
$$
\epsilon = \epsilon(f^A, H^{AB}).
$$
The various forms of the stress
tensor\footnote{The sign of the stress tensor here is the opposite of that of
$\sigma_{AB}$ defined in
\cite[eq. (3.5)]{ABS}, but agrees with the usage in
\cite[equation (4.2)]{BeSch03}.}
are defined from the
stored energy function via
\begin{align*}
\tau_{AB} &= + 2 \frac{\partial \epsilon}{\partial H^{AB}} \,, \quad
\tau_{ij} = n f^A{}_{i} f^B{}_{j} \tau_{AB} \,, \\
\quad \tau_i{}^A &=
f^B{}_{i} \tau_{BC} H^{CA} \,, \quad \tau_A{}^i = n \tau_{AB} f^B{}_j \delta^{ji}
\,.
\end{align*}
We have
\begin{align*}
\frac{\partial \epsilon}{\partial (f^A{}_i)} &= n^{-1} \tau_A{}^i
, \\
\frac{\partial \epsilon}{\partial (\phi^i{}_A)} &= - \tau_i{}^A
.
\end{align*}

The elasticity tensor $L_i{}^A{}_k{}^B$ is defined by
\begin{equation}\label{eq:elast-L}
L_i{}^A{}_k{}^B = \frac{\partial \tau_i{}^A}{\partial (\phi^k{}_B)}.
\end{equation}
It follows from this definition and the assumptions above that
the elasticity tensor
has the symmetries
\leqn{Lprop1}{
L^{iAjB}=L^{jBiA}= L^{iABj} = L^{AijB}
}
where
\eqn{Lprop2}{
L^{iAjB} = \delta^{il}\delta^{jk} L_{l}^A{}_{k}{}^{B}.
}
Clearly, we also have
$$
\partial_A \tau_i{}^A = L_i{}^A{}_k{}^B \partial_A \partial_B \phi^k.
$$

\subsection{Variational formulation} \label{sec:variational}
We derive the field equations for a self-gravitating elastic body from the
elastic action, supplemented by a term giving
Newton's force law and Newton's law of
gravitation. The Lagrange density is of
the form $\Lambda \epsilon_{0123}$ where $\epsilon_{0123}$ is the
4-volume element on spacetime $\Re \times \Re^3_{\Sp}$ in Cartesian
coordinates, and
$$
\Lambda = \Lambda^{grav} + \Lambda^{pot} + \Lambda^{kin} + \Lambda^{elast}
$$
where
\begin{align*}
\Lambda^{grav} &= \frac{|\nabla U|^2}{8\pi G} ,\\
\Lambda^{pot} &= \rho U \chi_{f^{-1}(\Bo)}  ,\\
\Lambda^{kin} &= \half \rho v^2  \chi_{f^{-1}(\Bo)}  , \\
\Lambda^{elast} &= - n \epsilon \chi_{f^{-1}(\Bo)}  \,,
\end{align*}
with $\chi_{f^{-1}(\Bo)}$ given by \eqref{eq:chidef},
$$
|\nabla U|^2 = \partial_i U \partial_j U \delta^{ij} ,
$$
and
$$
v^2 = v^i v^j \delta_{ij} .
$$
\subsubsection{Eulerian picture}
The action in the Eulerian picture then takes the form
$$
\LL = \int \Lambda \epsilon_{0123} dx^0 dx^1 dx^2 dx^3
$$
with $\Lambda = \Lambda(U, \partial_i U , f^A, f^A{}_0, f^A{}_i)$.
The Euler-Lagrange equations
are of the form $\EE_A = 0$, $\EE_U = 0$ with
\begin{align*}
- \EE_A &= \partial_\mu \frac{\partial \Lambda}{\partial (\partial f^A{}_\mu)}
- \frac{\partial \Lambda}{\partial (f^A)} \,, \\
- \EE_U &= \partial_i \frac{\partial \Lambda}{\partial (\partial_i U)}
- \frac{\partial \Lambda}{\partial U} \,.
\end{align*}
We have
$$
- \EE_U = \frac{\Delta U }{4\pi G} - \rho \chi_{f^{-1} (\Bo)} .
$$

Next, we consider the Euler-Lagrange terms generated by variations
with respect to the configuration $f^A$. A calculation using \eqref{eq:nphi}
shows that the factor $\rho \chi_{f^{-1}(\Bo)}$ gives no contribution to the
Euler-Lagrange equations, and hence the kinetic term in the action gives
\begin{align*}
- \EE_A^{kin} &=  \partial_\mu \frac{\partial \Lambda_{kin}}{\partial
  (f^A{}_\mu)}
- \frac{\partial \Lambda^{kin}}{\partial (f^A)} \\
&= [
\partial_k ( \half \rho \phi^k{}_A v^2 )
  - \partial_t (\rho \delta_{mn} v^m \phi^n{}_A )
  - \partial_k (\rho \delta_{mn} v^m v^k \phi^n{}_B)
] \chi_{f^{-1}(\Bo)}
,
\end{align*}
which after some calculations, using (\ref{eq:continuity}) and
(\ref{eq:nphi}), gives
\begin{align*}
-\EE_A^{kin} &= - \rho v^\mu \partial_\mu v^m  \delta_{mn}  \phi^n{}_A
\chi_{f^{-1}(\Bo)}
\\
\intertext{use (\ref{eq:HvmunufA})}
&= \rho H_{AB} v^\mu v^\nu f^A
\chi_{f^{-1}(\Bo)}
.
\end{align*}
The elastic term gives, using (\ref{eq:nphi})
and \eqref{eq:ddchi},
\begin{align*}
- \EE_A^{elast} &= \partial_i \frac{\Lambda^{elast}}{\partial (f^A{}_i)}
 - \frac{\partial \Lambda^{elast}}{\partial(f^A)}
\\
&=
 - \partial_i ( \tau_A{}^i \chi_{f^{-1}(\Bo)} )
.
\end{align*}
In view of \cite[lemma 2.2]{ABS}, we have that the divergence
$\partial_i (\tau_A{}^i
\chi_{f^{-1}(\Bo)} ) $ is integrable only if the zero traction boundary
condition
$$
\tau_A{}^j n_j \big{|}_{\partial f^{-1}(\Bo)} =0
$$
holds, in which case the identity
$$
\partial_i ( \tau_A{}^i \chi_{f^{-1}(\Bo)} ) = \partial_i  \tau_A{}^i  \chi_{f^{-1}(\Bo)}
$$
is valid.

Finally, the potential term gives, using \eqref{eq:nphi} and \eqref{eq:ddchi},
\begin{align*}
- \EE_A^{pot} &= \partial_i \frac{\Lambda^{pot}}{\partial (f^A{}_i)} -
\frac{\partial \Lambda^{pot}}{\partial f^A} \\
&= \partial_i (\rho \phi^i{}_A U \chi_{f^{-1}(\Bo)})
- \phi^i{}_A \partial_i \chi_{f^{-1}(\Bo)} \\
&= \rho \phi^i{}_A \partial_i U \chi_{f^{-1}(\Bo)}
.
\end{align*}
Adding the terms, we have
\begin{equation} \label{eq:EEA}
- \EE_A = \rho H_{AB} v^\mu v^\nu \partial_\mu \partial_\nu f^A - \partial_i \tau_A{}^i + \rho
  \phi^i{}_A \partial_i U \quad \text{ in } f^{-1}(\Bo) ,
\end{equation}
subject to the boundary condition
$$
\tau_A{}^j n_j \big{|}_{f^{-1}(\Bo)} = 0 .
$$
Hence,
we find that
the Euler-Lagrange equations $\EE_A = 0, \EE_U = 0$ are equivalent to the
system
\begin{subequations} \label{eq:euler}
\begin{align}
- \rho H_{AB} v^\mu v^\nu \partial_\mu \partial_\nu f^B + \partial_i
\tau_A{}^i &= \rho \phi^i{}_A \partial_i U ,
\quad \text{ in }
f^{-1} (\Bo) \label{eq:dyn-eul} \\
\Delta U &= 4\pi G \rho \chi_{f^{-1}(\Bo)}, \\
\tau_A{}^i n_i \big{|}_{\partial f^{-1}(\Bo)} &= 0\label{eq:bound-eul}
\end{align}
\end{subequations}
where $n_i$ is the unit outward pointing normal to $\partial f^{-1}(\Bo)$.
These are the field equations for a dynamical elastic body in Newtonian
gravity displayed in the Eulerian frame. We note that, by using (\ref{eq:HvmunufA}) and
multiplying by $f^A{}_i$, equations (\ref{eq:dyn-eul}) and
(\ref{eq:bound-eul}) take the form
\begin{equation}\label{eq:dyn-eul-vv}
\rho  v^\mu \partial_\mu v_i  + \partial_j
\tau_i{}^j = \rho \partial_i U
\quad \text{ in }
f^{-1} (\Bo) \AND \tau_i{}^j n_j \big{|}_{\partial f^{-1} (\Bo)} = 0 .
\end{equation}
\subsubsection{Material frame}
Similarly, the action in the material frame is given by
\begin{align*}
\LL^{mat} &= \int \phi^* (\Lambda \epsilon_{0123}) dX^0 dX^1 dX^2 dX^3
\\
&= \int J \Lambda^{mat} dX^0 dV_{ABC} dX^A dX^B dX^C
\end{align*}
where
$$
J = (\phi^*n)^{-1} = \det(\del_A\phi^i)
$$
is the Jacobian of $\phi$, and
$\Lambda^{mat} = \Lambda^{mat}(U, \phi^i, \phi^i_0, \phi^i{}_A)$
is given by the relation
\begin{equation}\label{eq:JLammat}
J \Lambda^{mat} = J \frac{|\nabla \bar U|^2_H}{8\pi G} +
(m \bar U + \half m
v^2 - \bar \epsilon) \chi_{\Bo}
\end{equation}
where
$\chi_{\Bo}$ is the indicator function of the support of the body and
$\bar U, \bar \epsilon$ are defined along the lines of
\cite{ABS}. In particular, $\bar U = U \circ \phi$, and
$$
|\nabla \bar U|^2_{H} = H^{AB} \partial_A \bar U \partial_B \bar U
$$
is the pullback to $\Re^3_{\Bo}$ of $|\nabla U|^2$ where we recall that
\eqn{Hfdef}{H^{AB} = f^A{}_i f^B{}_j \delta^{ij} \AND f_i{}^A = (\partial_A \phi^i )^{-1}.
}
Similarly, $\bar \epsilon = \epsilon \circ \phi$ so that $\bar \epsilon =
\bar \epsilon (f^A{}_i,H^{AB})$. We note also that
$$
\det (H^{AB}) = \det (f^A{}_i)^2 = J^2 .
$$

The Euler-Lagrange equations in material frame are
$\EE_i = 0, \EE_{\bar U} = 0$ where
\begin{align*}
- \EE_i &= \partial_A \frac{\partial \LL^{mat}}{\partial (\phi^i{}_A)} -
  \frac{\partial \LL^{mat}}{\partial (\phi^i)}, \\
- \EE_{\bar U} &= \partial_A \frac{ \partial \LL^{mat}}{\partial (\partial_A
  \bar U)} - \frac{\partial \LL^{mat}}{\partial (\bar U)}
.
\end{align*}
This gives the system of equations
\begin{subequations}
\lalign{eq:material}{
-m \partial_t^2 \phi^i + \partial_A (\bar{\tau}^i{}^A) &= m \delta^{ij}f^A_j\del_A \Ub \quad \text{\rm in } \Bo,   \label{eq:material.1} \\
\Delta_{H} \Ub &= 4 \pi G J^{-1} m \chi_{\Bo}\quad \text{\rm in }\Re^3_{\Bo} , \label{eq:material.2} \\
\nu_A\bar{\tau}^i{}^A |_{\partial\Bo} & =0  \label{eq:material.3}
}
\end{subequations}
where $\nu_A$ is the unit outward pointing normal to $\del\Bc$,
$$
\bar{\tau}^i{}^A = \delta^{ik}J (f^A{}_{j} \tau_k{}^j )\circ \phi
$$
is the Piola transform of $\tau_i{}^j$, and
\eqn{ebdef}{
H = H_{AB} dX^A dX^B = \phi^*(\delta_{ij}dx^i dx^j)
}
is the pull back of the Euclidean metric $\delta_{ij}dx^i dx^j$ under the map $\phi$. Observe that in
\eqref{eq:material.2},
we could also have used the
notation
$\bar \rho = J^{-1} m$ since $J^{-1} = n \circ \phi$.

Rescaling time and the Newtonian potential, we can write the evolution equations \eqref{eq:material.1}-\eqref{eq:material.3} in the
form
\begin{subequations}
\lalign{evolve}{
-\del_t^2 \phi^i + \del_A \taub^{Ai} - \ep^2 \delta^{ij}f^A_j\del_A\Ub & = 0 \quad \text{\rm in } \Bc, \label{evolve:1} \\
\Delta_{H} \Ub & = J^{-1} m \chi_{\Bc} \quad \text{\rm in }\Re^3_{\Bc}  , \label{evolve:2} \\
\nu_A \taub^{Ai}|_{\del \Bc} & = 0 \label{evolve:3}
}
\end{subequations}
where
\eqn{epdef}{
\ep^2 = 4\pi m G.
}
We note that the Poisson equation \eqref{evolve:2} can be written out more explicitly as
\leqn{PoissonC}{
\del_A\bigl(J H^{AB} \del_B \Ub \bigr)  = m \chi_{\Bc}.
}

It follows from the symmetry properties \eqref{Lprop1} of the elasticity tensor that
\leqn{Lbdef}{
\Lb^{iA}{}_{j}{}^{B} = \frac{\del \taub^{iA}}{\del \del_B\phi^j}
}
satisfies
\leqn{Lbprop1}{
\Lb^{iAjB}=\Lb^{jDiB}= \Lb^{iABj} = \Lb^{AijB}
}
where
\eqn{LBprop2}{
\Lb^{iAjB} = \delta^{jk} \Lb^{iA}{}_{j}{}^{B}.
}

Since we will be working in the material representation for the remainder of the article, we
will drop the $\Bc$ from the body space $\Rbb_{\Bc}$ and denote it simply by $\Rbb^3$.
For latter use, we define the following nonlinear
functionals:
\leqn{funcs}{
E^i(\phi) = \del_A\bigl(\taub^{iA}(\del\phi) \bigr), \AND
\EdB^i(\phi) =    \tr_{\del\Bc}\nu_A \taub^{iA}(\del\phi)
}
where $\del\phi = (\del_A \phi^i)$.

\subsect{assume}{Constitutive conditions}
We make the following assumptions on the elastic material:
\begin{enumerate}
\item $\taub^{iA}$ is a smooth function of its arguments $(\del\phi)$ in the neighborhood
of the identity map
\eqn{psi0def}{
\psi_0^i \: : \: \Rbb^3 \longrightarrow \Rbb^3 \: : \: (X^i) \longmapsto (\psi_0^i(X) = X^i),
}
\item \label{point:2}
the identity map
is an equilibrium solution to \eqref{evolve:1}-\eqref{evolve:3} for $\ep=0$, i.e.
\eqn{equilib}{
\bigl(E^i(\psi_0),\EdB^i(\psi_0)\bigr) = (0,0),
}
and
\item
\eqn{adef}{
a^{iA}{}_{j}{}^{B} = \Lb^{iA}{}_{j}{}^{B}\Bigl|_{\del\phi=\del\psi_0}
}
satisfies the following properties:
\begin{enumerate}
\item there exists an $\omega > 0$ such that
\eqn{aprop3}{
a^{iAjB}(X)\xi_i\xi_A\eta_j\eta_B \geq \omega |\xi|^2 |\eta|^2
}
for all $\xi,\eta \in \Rbb^3$, and
\item there exists a $\gamma >0$ such that
\eqn{aprop4}{
\Half a^{iAjB}\sigma_{iA}\sigma_{jB} \geq \gamma |\sigma|^2
}
for all $\sigma=(\sigma_{iB})\in \Mbb_{3\times 3}$
with $\sigma_{iB} = \sigma_{Bi}$.
\end{enumerate}
\end{enumerate}
We remark that condition (\ref{point:2}) above
is satisfied for a stored energy function
which has
a minimum at some reference configuration. See
\cite[equations (3.20)-(3.22)]{ABS}.

\section{Construction of initial data}\label{compatsec}

In order to prove the existence of dynamical solutions to the evolution equations \eqref{evolve:1}-\eqref{evolve:3}, we first
need to construct initial data that satisfy the \emph{compatibility conditions} to a sufficiently high order. The compatability
conditions are defined as follows.
\begin{Def} \label{compat}
Fixing $s>3/2+1$, we say that the initial data
\eqn{compata}{(\phi^i|_{t=0},\del_t|_{t=0}\phi^i)=(\phi^i_0,\phi^i_1)\in W^{s+1,2}(\Bc,\Rbb^3)\times
W^{s,2}(\Bc,\Rbb^3)}
satisfies the \emph{compatibility conditions to order $r$} $(0\leq r \leq s)$ if there exists maps
\eqn{compatb}{
\phi^i_\ell \in W^{s+1-\ell,2}(\Bc,\Rbb^3) \quad \ell=2,3,\ldots,r
}
that satisfy
\alin{evolvecompat}{
\del_t^{\ell-2}\bigl(-\del_t^2 \phi^i + \del_A \taub^{Ai} - \ep^2 \delta^{ij}f^A_j\del_A\Ub\bigr)\bigl|_{t=0} & = 0 \quad \text{\rm in } \Bc,  \\
\del_t^\ell\bigl(\Delta_{H} \Ub - J^{-1} m \chi_{\Bc}\bigr)\bigl|_{t=0} &= 0\quad \text{\rm in }\Re^3  , \\
\del_t^\ell \bigl((\nu_A \taub^{Ai})|_{\del \Bc}\bigr)\bigl|_{t=0} & = 0,
}
for $\ell=0,1,2,\ldots r$ where, after formally differentiating, we set $\del_t^\ell|_{t=0} \phi^i = \phi^i_\ell$.
\end{Def}

\subsection{Potential theory} \label{sec:potential}
Before we can solve the problem of existence of
initial data satisfying the compatibility conditions and the time evolution of this data, we first need to
develop some potential theory. To begin, we set
$\Bc_{+} = \Bc$, $\Bc_{-} = \Rbb^3\setminus \overline{\Bc}$, and let
$$
E(X,Y) = - \frac{1}{4\pi |X-Y|}
$$
denote the Newton potential so that $\Delta E = \delta$.
We also let $\SS$ and $\DD$ denote
the single and double layer potentials
\begin{align*}
\SS [f] (X) &= \int_{\del\Bc} E(X,Y) f(Y)\, d\sigma(Y) \quad X \notin \del \Bc
\intertext{and}
\DD [f] (X) &= \int_{\del \Bc} \ddnuy E(X,Y) f(Y)\, d\sigma(Y) \quad X \notin \del \Bc,
\end{align*}
respectively, where $d\sigma$ is the induced surface measure on $\del \Bc$, and $\normal$ is the outward pointing normal.
Restricting to $\del \Bc$, we have for $X \in \del \Bc$ that
\begin{equation*}
\SS [f] \big{|}_{\del \Bc} (X)  = \int_{\del\Bc} E(X,Y) f(Y)\, d\sigma(Y) \AND \DD [f] \big{|}_{\del \Bc_{\pm}} (X) = \bigl(\pm \Half I + K\bigr)[f](X)
\end{equation*}
where
$$
K [f](X) = P.V. \int_{\del \Bc} \ddnuy E(X,Y) f(Y)\, d\sigma(Y).
$$
Further,
\begin{equation*}
\frac{\del\;}{\del\nu} \SS[f]\big{|}_{\del \Bc_{\pm}}  = \bigl(\pm I + K^*\bigr)[f]
\end{equation*}
where $K^*$ is the adjoint of $K$.
We recall the following well known relations between the boundary value problems for
$\Delta$ in $\Bc_{\pm}$ and these potentials:
\begin{itemize}
\item[(i)] The solution to the Dirichlet problem
\begin{equation*}\label{eq:dirich-q}
\Delta u = 0, \quad \tr_{\del \Bc} u = \psi
\end{equation*}
on $\Bc_{\pm}$ is given by
\begin{equation*}
u = \DD[f]
\end{equation*}
where f solves
\begin{equation*}
\bigl(\pm \Half I + K\bigr) [f] = \psi.
\end{equation*}
\item[(ii)] The solution to the Neumann problem
\begin{equation*}
\Delta u = 0, \quad \tr_{\del \Bc}\, \ddnu u = \psi
\end{equation*}
on $\Bc_{\pm}$ is given by
\begin{equation*}
u = \SS[f]+ C
\end{equation*}
where $f$ solves
\begin{equation*}
\bigl(\mp\Half + K^*)[f] = \psi,
\end{equation*}
and $C$ is an arbitrary constant. The solution exists if and only if $\int_{\del \Bc} \psi \, d\sigma = 0$.
\end{itemize}

For the moment, we consider the interior Dirichlet and Neumann problems on $\Bc_+ = \Bc$. The solution for the Dirichlet
problem has the property that
\begin{equation}\label{eq:Direst}
u \in W^{s,p}(\Bc) \quad \text{ if }  \tr_{\del \Bc} u \in
B^{s-1/p,p}(\del \Bc),
\end{equation}
while for the Neumann problem, we have
$$
u \in W^{s,p}(\Bc) \quad \text{ if } \tr_{\del\Bc} \ddnu u \in B^{s-1-1/p,p}(\del \Bc).
$$
These results are classical, see \cite{ADN:I}.
Defining the volume potential of a density $f$ in $\Bc$ to be
$$
\VV[f](X) = \Delta^{-1} (f \chi_{\Bc})(X) = \int_{\Bc} E(X,Y) f(Y)\, d^3 Y,
$$
we abuse notation and say that $\VV[f] \in W^{k,p}(\Bc)$ if
the restriction to $\Bc$ has this property, and similarly for the other
potentials.
Using the fact that
$\partial_{X^A} E(X,Y) = - \partial_{Y^A} E(X,Y)$, we have
$$
\partial_{A}\VV[f](X) = -\int_{\Bc} \partial_{Y^A} E(X,Y) f(Y)\, d^3 Y.
$$
This gives, after a partial integration,
\begin{equation} \label{eq:dVV}
\partial_{A}\VV[f] =  \VV[\partial_{A} f] - \SS[ \tr_{\del \Bc} f
  \normal^A ].
\end{equation}

Let $u$ be the solution of the Dirichlet problem
$$
\Delta u = 0, \quad \tr_{\del \Bc} u = \psi.
$$
By Green's theorem, we have
\begin{align*}
\int_{\Bc} (\Delta_Y u(Y) E(X,Y) - & u(Y) \Delta_Y E(X,Y) )\, d^3 Y = \\
&\quad \int_{\del\Bc} \left (\ddnuy u(Y) E(X,Y) - u(Y) \ddnuy E(X,Y)\right ) d\sigma(Y).
\end{align*}
Since $u$ solves the Dirichlet problem with boundary data $\psi$,
and
\leqn{diffE}{
\Delta_Y E(X,Y) = \delta(X-Y),
} this gives
\begin{equation}\label{eq:DD}
\DD[u]|_{\Bc} = \SS[\tr_{\del\Bc} \ddnu u] + u.
\end{equation}
Upon taking the limit from the interior at $\del\Bc$, we have
$$
(\Half I + K)[u]|_{\del\Bc} = S[\tr_{\dO} \ddnu u] + u.
$$

Consider a metric $g_{AB}$ on $\Re^3$ with covariant derivative $\nabla_A$.
Let $\{e_a\}_{a=1,2}$ be a tangential frame on $\del\Bc$, and let
$h_{ab}$ denote the induced metric on $\del\Bc$ with
covariant derivative $D$.
Let a vector field $\xi$ be given.
Decompose $\xi$ into tangential and normal component at
$\del\Bc$,
$$
\xi = P^a e_a + \la \xi , \normal\ra \normal.
$$
Introduce a Gaussian foliation near $\del\Bc$. Then $g$ takes the form
$$
g_{AB} dX^A dX^B = dr^2 + h_{ab}(y,r)dy^a dy^b
$$
where $y^a$ are coordinates on $\del\Bc$.
Extending $\normal$ in a neighborhood of $\del\Bc$ using a Gaussian foliation,
we have
$$
\nabla_A \xi^A = h^{ab} \nabla_a \xi_b + (\nabla \xi)(\normal,\normal).
$$
The last term vanishes due to $\nabla_{\normal} \normal = 0$, which is valid in a
Gauss foliation. Then we have
$$
\nabla_A \xi^A = h^{ab} \nabla_a \xi_b = h^{ab} ( D_a P_b - \lambda_{ab} \la \xi , \normal \ra )
$$
where $\lambda_{ab} = \la \nabla_a \normal, e_b \ra = \Half \Lie_{\normal}
h_{ab}$ is the second fundamental form. Let $H = h^{ab} \lambda_{ab}$ denote
the mean curvature of $\del\Bc$. Then we have
$$
D_a P^a = H \la \xi , \normal \ra + \nabla_A \xi^A.
$$
Now specialize to the Euclidean case; let $g_{AB}= \delta_{AB}$
be the Euclidean metric
and let $\xi = \partial_{A}$ for some fixed $A$. Then $\nabla_A \xi^A = 0$, and
we have
\leqn{mean}{
D_a P^a = H \normal^A.
}
For $X \in \Bc$, we calculate using \eqref{diffE} and the divergence theorem that
\begin{align*}
\partial_{A} \SS[f](X) &= -\int_{\del\Bc} f(Y) P^a D_a E(X,Y)\, d\sigma(Y)
 \\
&\quad - \int_{\del\Bc} f(Y) \la \normal , \partial_{Y^A} \ra \ddnuy E(X,Y)\, d\sigma(Y) \\
&= \int_{\del\Bc} (D_a P^a  f(Y) + P^a D_a f ) E(X,Y)\, d\sigma(Y) \\
&\quad
- \int_{\del\Bc} f(Y) \la \normal , \partial_{y^A} \ra \ddnuy E(x,y)\, d\sigma(Y).
\end{align*}
Thus we have by the above result and \eqref{mean} that
\begin{equation}\label{eq:dSS}
\partial_{A} \SS[f] = \SS[ f H \normal^A + \partial_{A}^{\parallel} f]
- \DD[f \normal^A].
\end{equation}
\begin{prop} \label{prop:SD}
The operators $S, \Half I + K$ have the
mapping properties
\begin{subequations}\label{eq:KSmap}
\begin{align}
\Half I + K: & B^{k-1/p,p} (\dO) \to B^{k-1/p,p}(\dO) ,\\
S : & B^{k-1-1/p,p}(\dO) \to
B^{k-1/p,p} (\dO) ,
\end{align}
\end{subequations}
for $s = k-1/p$, $k \geq 1$, $k$ an integer.
The corresponding statements for the single and double layer potentials are
\begin{subequations} \label{eq:DDSS}
\begin{align}
\DD &:\  B^{k-1/p,p}(\dO) \to W^{k,p}(\Om) \\
\SS &:\  B^{k-1-1/p,p}(\dO) \to W^{k,p}(\Om)
\end{align}
\end{subequations}
for $k \geq 1$, $k$ integer.
\end{prop}
\begin{proof}
We use induction to reduce the statement to the case $k=1$. This case follows
from \cite{FMM}, see also \cite{mitrea:taylor:besov}.
Suppose then that we have proved the statement for $k-1$.
To do the induction, assume $f \in
B^{k-1/p,p}(\dO)$. Let $u$ be the solution
to the Dirichlet problem with boundary data $f$. Then $u \in W^{k,p}(\Om)$,
and setting
 $\psi = \tr_{\dO} \ddnu
u$, we have $\psi \in B^{k-1-1/p,p}$. Equations
(\ref{eq:DD}) and (\ref{eq:dSS}) give
\begin{align*}
\partial_{x^i} \DD[f] &=
\partial_{x^i} \SS[\psi] + \partial_{x^i} u \\
&= \SS[\psi H \normal^i + \partial_{x^i}^{\parallel} \psi] - \DD[\psi \normal^i]
+ \partial_{x^i} u
\end{align*}
which by the induction assumption is in $W^{k-1,p}$. It follows that $\DD[f]
\in W^{k,p}(\Om)$ and hence $(\Half I + K) f \in B^{k-1/p,p}(\dO)$. This
proves the statement for for $\DD$ and $(\Half I + K)$ at regularity $k$.

Next, for $\SS$, we use the equation (\ref{eq:dSS}) for $f \in
B^{k-1-1/p,p}(\dO)$. Using the induction assumption and the statement for
$\DD$ and $(\Half I + K)$ just proved, we have
$$
\partial_{x^i} \SS[f] \in W^{k,p},
$$
which gives $\SS[f] \in W^{k,p}(\dO)$, and $S[f] \in B^{k-1/p,p}(\dO)$. This
completes the induction and the result follows.
\end{proof}

\begin{ex} \label{ex:1,p}
Let $f \in W^{1,p}(\Om)$. Then $\tr_{\dO} f \in
B^{1-1/p,p}$ and $\VV[\partial_{x^i} f] \in W^{2,p}(\Om)$. Further,
$\SS[\tr_{\dO} f \normal^i]$ solves a Dirichlet problem with boundary data
$S[\tr_{\dO} f \normal^i] \in B^{2-1/p,p}(\dO)$, and hence $\SS[\tr_{\dO} f
  \normal^i] \in W^{2,p}$. It follows, in view of (\ref{eq:dVV}),
that $(\VV f)\chi_{\Om} \in
W^{3,p}(\Om)$.
\end{ex}

\begin{ex} Let $f \in W^{2,p}(\Bc)$. We have
$$
\partial_{x^i} \VV[f] = \VV[\partial_{x^i} f] + \SS[\tr_{\dO} f \normal^i].
$$
Since $\partial_{x^i} f \in W^{1,p}(\Bc)$, we have from
Example \ref{ex:1,p} that $\VV[\partial_{x^i} f] \in W^{3,p}(\Bc)$. Further,
$\tr_{\dO} f \in B^{2-1/p,p}(\dO)$ and hence $\SS[\tr_{\dO} f \normal^i] \in
W^{3,p}(\Bc)$. Therefore $\VV[f] \in W^{4,p}(\Bc)$.
\end{ex}
\begin{prop} \label{invLapmap1} Let $k \geq 1$, and assume
$f \in W^{k,p}(\Bc)$. Then $\VV[f] \in W^{k+2,p}(\Bc)$.
\end{prop}
\begin{proof}
The proof proceeds by induction, with base case $k=1$. For this case, the
statement follows by the argument in Example \ref{ex:1,p}. Suppose we have
proved the statement for $k-1$. We will make use of
the identity (\ref{eq:dVV}). By induction, $\VV[\partial_{x^i} f] \in
W^{k+1,p}(\Omega)$. Further, $\tr_{\dO} f \in B^{k-1/p}$ and hence by
(\ref{eq:KSmap}), $S[\tr_{\dO} f \normal^i] \in B^{k+1-1/p}(\dO)$. It follows by
(\ref{eq:Direst}) that $\SS[\tr_{\dO} f \normal^i] \in W^{k+1,p}(\Bc)$. This
shows that $\partial_{x^i} \VV[f] \in W^{k+1,p}(\Bc)$ and hence $\VV[f] \in
W^{k+2,p}(\Bc)$.
\end{proof}

Similar arguments combined with the mapping properties \eqref{Deltaiso}-\eqref{invLap} of the
Laplacian on the weighted Sobolev spaces can be used to establish the following
proposition for the volume potential on $\Bc_{-}$.
\begin{prop} \label{invLapmap2} Let $k \geq 1$, $-1<\delta<0$, and assume
$f \in W^{k,p}_{\delta-2}(\Bc_-)$. Then $\VV[f] \in W^{k+2,p}_\delta(\Bc_-)$.
\end{prop}

\subsect{newt}{The Poisson equation in the material frame} The next step in solving the problems of the existence of
initial data satisfying the compatibility conditions and the time evolution of this data is to
establish a number of smoothness properties for solutions to the Poisson equations in the
material frame. We begin by defining the spaces
\eqn{Wcdef}{
\Wc^{k,s,p}_\delta(\Rbb^3) = \bigl\{ \, u \in W^{k,p}_\delta(\Rbb^3,V)\, \bigl| \; u|_{\Bc} \in W^{k+s,p}(\Bc)
\AND  u|_{\Bc_-} \in W^{k+s,p}_\delta(\Bc_{-})\, \bigr\}
}
for
$1<p< \infty$, $k\in \Zbb$ and $\delta \in \Rbb$.
It is not difficult to verify that these spaces are complete with respect to the norm
\leqn{norm}{
\norm{u}_{\Wc^{k,s,p}_\delta(\Rbb^3)} = \norm{u|_{\Bc}}_{W^{k+s,p}(\Bc)}+ \norm{u|_{\Bc_-}}_{W^{k+s,p}_\delta(\Bc_-)} +
\norm{u}_{W^{k,p}_\delta(\Rbb^3)},
}
and hence Banach spaces.

\begin{thm} \label{newtA}
Suppose $1<p<\infty$, $s\in \Zbb_{\geq 0}$, $s+1>3/p$, and $-1 < \delta < 0$. Then there exist an open neighborhood $\Oct^{s+2,p}\subset W^{s+2,p}(\Bc,\Rbb^3)$
of $\psi_0$, and an analytic
map
\eqn{newtA1}{
\Ub \: :\: \Oct^{s+2,p} \longrightarrow \Wc^{2,s,p}_\delta(\Rbb^3) \: : \: \phi \longmapsto \Ub(\phi)
}
such that
\begin{enumerate}
\renewcommand{\theenumi}{(\roman{enumi})}
\renewcommand{\labelenumi}{(\roman{enumi})}
\item \label{point:i}
$\Ub(\phi)$ satisfies the Poisson equation \eqref{evolve:2} on $\Rbb^3$,
\item \label{point:ii}
for each $\phi\in \Oct^{s+2,p}$, the map $\phit = \psi_0 + \Ebb_\Bc\bigl(\phi-\psi_0|\Bc\bigr)$ is a $C^1$ diffeomorphism on $\Rbb^3$
that satisfies
$\phit-\psi_0 \in W^{2+s,p}_{-10}(\Rbb^3,\Rbb^3) \subset
C^1_{-10}(\Rbb^3,\Rbb^3)$ and
$\phit^{-1}-\psi_0 \in W^{2+s,p}_{-10}(\Rbb^3,\Rbb^3)$,
\item \label{point:iii}
$U = \Ub(\phi)\circ \phit^{-1} \in W^{2,p}_{\delta}(\Rbb^3)$ satisfies the Poisson equation $\Delta U = m\rho\chi_{\phi(\Bc)}$ on $\Rbb^3$
and $U(x)=\text{\rm o}(|x|^\delta)$ as $|x|\rightarrow \infty$, and
\item \label{point:iv}
for $k\in \Zbb_{\geq 1}$ and $k \leq s+1$, the derivative of
  $\Ub$ can be extended to
act on $W^{k,p}(\Bc)$, and moreover, the map
\eqn{Gmap9a}{
\Oct^{s+2,p} \ni \phi \longmapsto  D\Ub(\phi) \in L(W^{k,p}(\Bc))
}
is well defined and analytic\footnote{For a Banach space $X$, $L(X)$ denotes the set of continuous linear operators on $X$}.
\end{enumerate}
\end{thm}
\begin{proof}
\textbf{\ref{point:i}}
Fix $1<p<\infty$, $s+1>3/p$, and $-1 < \delta < 0$. Given $\psi \in W^{2+s,p}(\Bc)$, we define
\eqn{phitdef}{
\phit = \psi_0 + \Ebb_{\Bc}(\psi), \quad
\phi^i_{A} = \del_A \phit^i, \AND
(f^A_i)  = (\phi^i_A)^{-1}.
}
Since matrix inversion and the determinant both define analytic maps in a neighborhood of the identity, it follows
from \eqref{wmlem}, proposition 3.6 of \cite{Heil}, the continuity of extension and differentiation,
the analyticity of continuous linear maps, and the property that the composition of analytic maps are again analytic that there exists a $R>0$ such that the maps
\leqn{detJmap}{
B_R(W^{2+s,p}(\Bc,\Rbb^3))\ni \psi \longmapsto \det(\phi^i_A)-1 \in W^{1+s,p}_{-11}(\Rbb^3)
}
and
\leqn{invJmap}{
B_R(W^{2+s,p}(\Bc,\Rbb^3))\ni \psi \longmapsto (f^A_i - \delta^A_i) \in W^{1+s,p}_{-11}(\Rbb^3,\Mbb_{3\times 3} )
}
are well defined and analytic.
Recalling that
$H^{AB} = f^{A}_i \delta^{ij} f^{B}_j$,
we see from the same arguments that the map
\leqn{mpb}{
B_R(W^{2+s,p}(\Bc,\Rbb^3))\ni \psi \longmapsto J H^{AB} -\delta^{AB} \in W^{1+s,p}_{-11}(\Rbb^3,\Mbb_{3\times 3} )
}
is analytic where $J=\det(\phi^i_A)$. Using the multiplication inequalities \eqref{mlem} and \eqref{wmlem}, we find that
\lalign{PoissonmapABC}{
\norm{\del_A\bigl(J H^{AB} \del_B \Ub \bigr)\bigl|_{\Bc}}_{W^{s,p}(\Bc)} & \lesssim \norm{\psi}_{W^{2+s,p}
(\Bc,\Rbb^3)}\norm{\Ub|_{\Bc}}_{W^{2+s,p}(\Bc)},
\label{PoissonmapA}\\
\norm{\del_A\bigl(J H^{AB} \del_B \Ub \bigr)\bigl|_{\Bc_-}}_{W^{s,p}_\delta(\Bc_-)} &
\lesssim \norm{\psi}_{W^{2+s,p}(\Bc,\Rbb^3)}\norm{\Ub|_{\Bc_-}}_{W^{2+s,p}_\delta(\Bc_-)}
\label{PoissonmapB}
\intertext{and}
\norm{\del_A\bigl(J H^{AB} \del_B \Ub \bigr)}_{W^{0,p}_{\delta-2}(\Rbb^3)} &\lesssim \norm{\psi}_{W^{2+s,p}(\Bc,\Rbb^3)}\norm{\Ub}_{W^{2,p}_\delta(\Rbb^3)}
\label{PoissonmapC}.
}
From the analyticity of the map \eqref{mpb}, the bilinear estimates \eqref{PoissonmapA}-\eqref{PoissonmapC}, the analyticity of continuous bilinear maps, and the property that the composition of analytic maps are again analytic, it follows
that the
\leqn{Poissonmap}{
B_R(W^{2+s,p}(\Bc,\Rbb^3))\times \Wc^{2,s,p}_{\delta}(\Rbb^3) \ni (\psi,\Ub) \longmapsto
\del_A\bigl(J H^{AB} \del_B \Ub \bigr) \in \Wc^{0,s,p}_{\delta-2}(\Rbb^3)
}
is well defined and analytic. Together, \eqref{Deltaiso}, and Propositions \ref{invLapmap1} and \ref{invLapmap2} imply that
\leqn{NsmoothA}{
\Delta^{-1}(\chi_\Bc) \in \Wc^{2,s,p}_{\delta}(\Rbb^3),
}
and the Laplacian
\leqn{NsmoothB}{
\Delta \: :\: \Wc^{2,s,p}_{\delta}(\Rbb^3) \longrightarrow \Wc^{0,s,p}_{\delta-2}(\Rbb^3)
}
is an isomorphism with inverse given by \eqref{invLap}.

From \eqref{Poissonmap}, \eqref{NsmoothA}, and \eqref{NsmoothB}, we see that
\begin{align*}
F  &: B_R(W^{2+s,p}(\Bc,\Rbb^3))\times \Wc^{2,s,p}_{\delta}(\Rbb^3)
\longrightarrow W^{2,s,p}_{\delta}(\Rbb^3) ,
 \\
&\quad (\psi,\Ub) \longmapsto \Delta^{-1}\Bigl(\del_A\bigl(J H^{AB} \del_B \Ub \bigr)-m\chi_\Bc\Bigr)
\end{align*}
is well defined and analytic. Evaluating $F$ at $\psi=0$ gives
\eqn{Fmap2}{
F(0,\Ub) = \Delta^{-1}\Bigl(\Delta \Ub - m\chi_\Bc\Bigl),
}
which shows that
\leqn{Fmap3}{
\Ub_0 = m \Delta^{-1}(\chi_{\Bc}) \in W^{2,s,p}_{\delta}(\Rbb^3)
}
satisfies
\leqn{Fmap4}{
F(0,\Ub_0) = 0.
}
Also, by the linearity of $F$ in its second argument and the invertibility
of the Laplacian, it is clear that
\leqn{Fmap5}{
D_2F(0,\Ub)\cdot \delta\Ub = \Delta^{-1}\Bigl(\Delta \delta\Ub\Bigr) = \delta \Ub.
}
Results \eqref{Fmap4} and \eqref{Fmap5} allow us to apply an analytic version of the implicit function
theorem (see \cite{Deim}, theorem 15.3) to conclude the existence of a unique analytic map, shrinking $R$ if necessary,
\leqn{Ubmap1}{
\Ub\: : \: B_R(W^{2+s,p}(\Bc,\Rbb^3))\longrightarrow W^{2,s,p}_{\delta}(\Rbb^3)
}
that satisfies
\eqn{Ubmap2}{
\Ub(0) = \Ub_0,
}
and
\eqn{Ubmap3a}{
F(\psi,\Ub(\psi)) = 0 \quad \forall \; \psi \in B_R(W^{2+s,p}(\Bc,\Rbb^3)).
}
From the definition of $F$ and the invertibility of $\Delta^{-1}$, it then follows that $\Ub(\psi)$ satisfies
\leqn{Ubmap3}{
\del_A\bigl(J H^{AB} \del_B \Ub(\psi) \bigr)= m\chi_\Bc.
}

\bigskip

\noindent
\textbf{\ref{point:ii} \& \ref{point:iii}}
Following Cantor \cite{Can75}, we consider the following group of diffeomorphisms on $\Rbb^3$
\eqn{Ddef}{
\Dc^{s,q}_\delta(\Rbb^3) := \bigl\{\, \phi : \Rbb^3 \rightarrow \Rbb^3\,|\, \text
{$\phi-\psi_0 \in W^{s,q}_\eta(\Rbb^3,\Rbb^3)$,  and $\phi^{-1}-\psi_0\in W^{s,q}_\eta(\Rbb^3,\Rbb^3)$} \,\bigr\}
}
where  $s>3/q+1$ and $\eta\leq 0$. Fixing $\psi \in B_R(W^{2+s,p}(\Bc,\Rbb^3))$, we get from
\eqref{detJmap} and \eqref{invJmap} that
\eqn{phitA}{
\phit=\phi(\psi) \in \Dc^{2+s,p}_{-10}(\Rbb^3).
}

Defining,
\leqn{UsidefA}{
U := U(\psi)\circ \phit^{-1},
}
we can apply Corollary 1.6 of \cite{Can75} to get\footnote{In \cite{Can75}, Cantor required that $\delta \leq -3/2$
because that was what he needed to prove the weighted multiplication inequality \eqref{wmlem}. It is clear that
his proofs are valid whenever the multiplication inequality holds and $W^{k,p}_\delta \subset C^1_b$. Consequently,
the only restriction on $\delta$ is that $\delta\leq 0$.
}
\leqn{UpsidefB}{
U \in W^{2,p}_\delta(\Rbb^3).
}
A straightforward calculation using the chain rule and \eqref{Ubmap3}, \eqref{UsidefA}, and \eqref{UpsidefB} then shows that
\eqn{UpsidefC}{
\Delta U = m \det(D(\phit^{-1})) \chi_{\phit(\Bc)},
}
while the fall off condition $U(x)=\text{\rm o}(|x|^\delta)$ as $|x|\rightarrow \infty$ follows from the weighted Sobolev inequality \eqref{wSob}.

\bigskip

\noindent

\textbf{\ref{point:iv}}
To begin, we assume that $k=1$ and observe that for $\theta^i \in W^{1,p}(\Bc)$ and $\Ub \in \Wc^{2,s,p}(\Rbb^3)$
\lalign{Gmap0}{
\norm{\Ebb_\Bc(\del_C\theta^i)\del_B\Ub}_{L^p_{\delta-1}(\Rbb^3)}
&\leq \norm{\del_C\theta^i\del_B\Ub}_{L^p(\Bc)}+\norm{\Ebb_\Bc(\del_C\theta^i)\del_B\Ub}_{L^p_{\delta-1}(\Bc_-)}
\notag \\
& \lesssim \norm{\theta}_{W^{1,p}(\Bc)}\norm{\Ub|_{\Bc}}_{W^{s+2,p}_\delta(\Bc)} +
\norm{\theta}_{W^{1,p}(\Bc)}\norm{\Ub|_{\Bc_-}}_{W^{s+2,p}_\delta(\Bc_-)} \notag \\
& \lesssim \norm{\theta}_{W^{1,p}(\Bc)}\norm{\Ub}_{\Wc^{2,s,p}(\Rbb^3)} \label{Gmap0.1}
}
where in deriving the result we have used property \eqref{ext2} of the extension operator $\Ebb_\Bc$,
the multiplication inequalities \eqref{mlem} and \eqref{wmlem}, and the assumption $1+s>3/p$.

Letting
\eqn{Htdef}{
H^{ABC}_i = \frac{\del JH^{AB}}{\del\phi^i_C}, 
}
we get, using the estimate \eqref{Gmap0.1} and the same arguments as above, that for $R$ small enough the map
\alin{Gmap1}{
G_A \: : \:   & B_R(W^{2+s,p}(\Bc,\Rbb^3))\times W^{1,p} (\Bc,\Rbb^3)\times \Wc^{2,s,p}_\delta(\Rbb^3) \times
W^{1,p}_\delta(\Rbb^3) \longrightarrow W^{-1,p}_{\delta-1}(\Rbb^3) \\
& \: :\: (\psi^i,\theta^i,\Ub,\Vb)
 \longmapsto J H^{AB} \del_B \Vb
+ H^{ABC}_i \Ebb_\Bc(\del_C\theta^i) \del_B(\Ub)
}
is analytic. From the continuity of differentiation and the trace map, we then
have that the map
\eqn{Gmap2}{
G \: : \:  B_R(W^{2+s,p}(\Bc,\Rbb^3))\times W^{1,p}(\Bc,\Rbb^3)\times \Wc^{2,s,p}_\delta (\Rbb^3) \times
W^{1,p}_\delta (\Rbb^3) \longrightarrow W^{-1,p}_{\delta-2}(\Rbb^3)
}
defined by
\eqn{Gmap2a}{
G(\psi^i,\theta^i,\Ub,\Vb) = \del_A G^A(\psi^i,\theta^i_B,\Ub,\Vb)
}
is analytic. Taking $\Ub_0$ as defined by \eqref{Fmap3}, a straightforward calculation and the invertibility of the Laplacian show that
\eqn{Gmap3}{
G(0,0,\Ub_0,0) = 0 \AND D_4 G(0,0,\Ub_0,0)\cdot \delta\Vb = \delta \Vb.
}
Therefore, we can again apply the analytic version of the implicit function
theorem (see \cite{Deim}, theorem 15.3) to conclude the existence of a unique analytic map, shrinking $R$ if necessary,
\leqn{Gmap4}{
\Vb\: : \: B_R(W^{2+s,p}(\Bc,\Rbb^3)) \times W^{1,p}(\Bc,\Rbb^3) \times
 \bigl(\Ub_0+B_{R}(\Wc^{2,s,p}_\delta) \bigr)\longrightarrow W^{1,p}_{\delta}(\Rbb^3)
}
that satisfies
\eqn{Gmap5}{
\Vb(0,0,\Ub_0,0) = 0,
}
and
\eqn{Gmap6}{
G(\psi^i,\theta^i_B,\Ub,\Vb(\psi^i,\theta^i_B,\Ub)) = 0
}
for all $(\psi^i,\theta^i,\Ub)$ $\in$
$B_R(W^{2+s,p}(\Bc,\Rbb^3))$ $\times$ $W^{1,p}(\Bc,\Rbb^3)$ $\times$
 $\bigl(\Ub_0+B_{R}(\Wc^{2,s,p}_\delta(\Rbb^3)) \bigr)$.

From the construction of $G$, and the uniqueness of the maps \eqref{Ubmap1} and \eqref{Gmap4}, it is not difficult to
verify that
\eqn{Gmap7}{
\Vb(\psi^i,\del_A\delta\psi^i,\Ub(\psi)) = D\Ub(\psi)\cdot \delta\psi \qquad \forall \: (\psi,\delta\psi)\in  B_R(W^{2+s,p}(\Bc,\Rbb^3))
\times W^{2+s,p}(\Bc,\Rbb^3).
}
From this and the density of $W^{s+1,p}(\Bc)$ in $W^{1,p}(\Bc)$, it follows that the derivative of $\Ub$ can be extended
act on $W^{1,p}(\Bc)$, and moreover, that the map
\leqn{Gmap8}{
B_R(W^{2+s,p}(\Bc,\Rbb^3))\ni \psi \longmapsto  D\Ub(\psi) \in L(W^{1,p}(\Bc))
}
is well defined and analytic. By \eqref{Ubmap1} above, we also have that the map
\leqn{Ubmap1aa}{
B_R(W^{2+s,p}(\Bc,\Rbb^3)) \ni \psi \longmapsto  D\Ub(\psi) \in L(W^{2+s,p}(\Bc))
}
is well defined and analytic. Together, the maps \eqref{Gmap8}-\eqref{Ubmap1aa} and interpolation imply that
the map
\eqn{Ubmap1bb}{
B_R(W^{2+s,p}(\Bc,\Rbb^3)) \ni \psi \longmapsto  D\Ub(\psi) \in L(W^{k,p}(\Bc))
}
is well defined and analytic for $k\in \Zbb$ and $1<k<s+2$.
\end{proof}

\begin{cor} \label{newtB}
The map
\eqn{newB1}{
\Lambda \: :\: \Oct^{s+2,p} \longrightarrow W^{s+1,p}(\Bc,\Rbb^3)
}
defined by
\eqn{newB2}{
 \Lambda^i(\phi) = -\delta^{ij}f^A_i\del_A \Ub(\phi)  \qquad
 \bigl((f^A_i)=(\del_A\phi^i)^{-1}\bigr)
}
is analytic. Moreover, for $k\in \Zbb_{\geq 1}$ and $k \leq s+1$, the derivative of  $\Lambda$ can be extended
act on $W^{k,p}(\Bc)$, and the map
\eqn{Gmap9}{
\Oct^{s+2,p} \ni \phi \longmapsto  D\Lambda(\phi) \in L(W^{k,p}(\Bc),W^{k-1,p}(\Bc))
}
is well defined and analytic.
\end{cor}
\begin{proof}
This follows directly from theorem \ref{newtA}, the multiplication inequality \eqref{wmlem}, proposition 3.6 of \cite{Heil}, and the fact that
compositions of analytic maps are again analytic.
\end{proof}

\subsect{leop}{The linearized elasticity operator}

We define the operator linearized elasticity and boundary operators by
\leqn{Adef}{
A(\phi)^i := \del_B\bigl(a^{iB}{}_j{}^D\del_D\phi^j\bigr)
}
and
\leqn{Abdef}{
\AdB(\phi)^i := \nu_Ba^{iB}{}_j{}^D\del_D\phi^j,
}
respectively. We also define
\eqn{YdefA}{
Y^{s,p} = W^{s,p}(\Bc,\Rbb^3)\times B^{s+1-1/p,p}(\del \Bc,\Rbb^3),
}
and for each $\phi \in W^{s+2,p}(\Bc,\Rbb^3)$,
\eqn{YdefB}{
Y_\phi^{s,p} = \{\, (b,t) \in Y^{s,p} \, | \, \Cc_1(b,t) = 0, \; \Cc_2(\phi,b,t) = 0 \, \}
}
where
\lalign{Ccdef1}{
\Cc_1(b,t) & = \int_\Bc b + \int_{\del\Bc} t, \label{Ccdef1.1}
\intertext{and}
\Cc_2(\phi,b,t) & = \int_\Bc b\times \phi + \int_{\del\Bc} t \times \phi. \label{Ccdef1.2}
}
Here, we are using the notation
\lalign{Ccdef2}{
(b\times \phi)^i & = \epsilon^i{}_{jk} b^j \phi^k, \label{Ccdef2.1}\\
\int_\Bc b & = \int_{\Bc} b\, d^3 X, \label{Ccdef2.2}
\intertext{and}
\int_{\del \Bc} t & = \int_{\del \Bc} t\, d\sigma.
}

For later use, we recall the following theorem from \cite{MH} concerning the surjectivity of the linearized elasticity operator.
\begin{thm}{\rm [\cite{MH}, theorem 1.11, Section 6.1]} \label{surjthm}
Suppose $1<p<\infty$, $s\geq 0$, and $\Pbb : Y^{s,p} \longrightarrow Y^{s,p}_{\psi_0}$ is any projection map.
Then the map
\eqn{surjthm1}{
\Abb \: :\: W^{s+2,p}(\Bc,\Rbb^3) \longrightarrow Y^{s,p}_{\psi_0} \: :\: \phi \longmapsto \Pbb(A(\phi),\AdB(\phi))
}
is surjective and
\eqn{surjthm2}{
\ker \Abb = \{ a+b\times X \, | \, a=(a^i),\; b=(b^i)\in \Rbb^3 \}
}
where $(b\times X)^i = \epsilon^i{}_{jA} b^j X^A$.
\end{thm}

\begin{rem} \label{transrem}
Letting
\eqn{transrem1}{
\tilde{X}^j = \frac{1}{\text{Vol}(\Bc)} \int_{\Bc} X^j \, d^3 X
}
denote the center of $\Bc$, a short calculation using the change of
coordinates
$\bar{X}^j = X^j-\tilde{X}^j$ and $\bar{\Bc} = \Bc - \tilde{X}$
shows that
\eqn{transrem3}{
\int_{\bar{\Bc}} \bar{X}^j \, d^3 \bar{X} =  \int_{\Bc} X^j - \tilde{X}^j  \, d^3 X.
}
Therefore, we can always arrange that
\leqn{transrem4}{
\int_{\Bc} X^j \, d^3 X = 0
}
by translating the domain $\Bc$. For the remainder of this article, we will always assume
that the condition \eqref{transrem4} holds.
\end{rem}

Next, for $p>3/2$, we define the spaces
\eqn{Xdef}{
X^{s,p} = \Bigl\{ \phi \in W^{s,p}(\Bc,\Rbb^3) \, \Bigl| \, \int_{\Bc} \phi = 0 \, \Bigr\},
}
and observe, using  Sobolev's and H\"{o}lder's inequalities,  that
\alin{BmapA}{
\Bigl|\int_\Bc \phi\times \psi \Bigr|  \leq \norm{\phi\times \psi}_{L^1(\Bc,\Rbb^3)}
& \lesssim \norm{\phi}_{L^\infty(\Bc,\Rbb^3)} \norm{\psi}_{L^1(\Bc,\Rbb^3)} \\
& \lesssim \norm{\phi}_{W^{2,p}(\Bc,\Rbb^3)} \norm{\psi}_{L^{p}(\Bc,\Rbb^3)} \\
& \lesssim \norm{\phi}_{W^{2+s,p}(\Bc,\Rbb^3)} \norm{\psi}_{W^{s,p}(\Bc,\Rbb^3)}
}
from which the continuity of the bilinear map
\eqn{BmapB}{
B \: : \: X^{s+2,p}\times X^s \longrightarrow \Rbb^3 \: :\: (\phi,\psi) \longmapsto \int_\Bc \phi\times \psi
}
follows. Setting
\eqn{UcdefA}{
B_\psi(\phi) := B(\phi,\psi),
}
we define for $p>3/2$ the following spaces
\eqn{Ucdef}{
\Uc^{s,p} = \{\, \psi \in X^{s,p} \, | \, \text{ $B_\psi \: : \: \ker A\cap X^{s+2,p} \rightarrow \Rbb^3$ is an isomorphism} \, \}.
}

\begin{lem} \label{Uclem}
$\psi_0 \in \Uc^{s,p}$ for all $s\geq 0$ and $3/2<p<\infty$.
\end{lem}
\begin{proof}
By \eqref{transrem4}, we have that
\eqn{Uclem1}{
\int_{Bc}\psi^i_0\, d^3 X = \int_{\Bc} X^i \, d^3 X = 0,
}
and
\eqn{Uclem2}{
\int_{\Bc} a + b\times X = \int_{\Bc} a + b\times \int_{\Bc} X = \text{Vol}(\Bc)
}
for all $a,b\in \Rbb^3$. Consequently,
\leqn{Uclem3}{
\psi_0 \in X^{s,p},
}
and
\leqn{Uclem4}{
\ker \Abb \cap X^{s+2,p} = \{\, b\times X \, | \, b\in \Rbb^3 \, \}
}
 by theorem \ref{surjthm}. Next,
\alin{Ulem7}{
B_{\psi_0}(b\times X) = \int_{\Bc} (b\times X)\times \psi_0
& = \int_{\Bc} (b\times X)\times X \\
& = \int_{\Bc} (X\cdot b)X-b|X|^2,
}
which, after taking the innerproduct with $a\in \Rbb^3$, yields
\leqn{Ulem8}{
b\cdot B_{\psi_0}(b\times X) = \int_{\Bc}(X\cdot b)(X\cdot a) - a\cdot b |X|^2.
}
The Cauchy-Schwartz inequality shows that
\leqn{Ulem9}{
(X\cdot b)^2 - |b|^2 |X|^2 \leq 0
}
and
\leqn{Ulem10}{
(X\cdot b)^2 - |b|^2 |X|^2 = 0 \quad \forall \; X \in \Bc \Longleftrightarrow b=0.
}
Combining \eqref{Ulem8}-\eqref{Ulem10}, we arrive at
\leqn{Ulem11}{
b \cdot B_{\psi_0} (b\times X) = 0 \Longleftrightarrow b=0.
}
By way of contradiction, suppose that the map
\leqn{Ulem12}{
\ker \Abb \cap X^{s+2,p} \longrightarrow \Rbb^3
}
is not surjective. Then the image $B_{\psi_0}(\ker \Abb \cap X^{s+2,p})$ is contained in a
two dimensional subspace, and therefore, there exists a non zero $a\in \Rbb^3$ such that
\eqn{Ulem13}{
a\cdot B_{\psi_0}(b\times X) = 0 \qquad \forall \: b \in \Rbb^3 .
}
But this is impossible by \eqref{Ulem11}, and hence the map \eqref{Ulem12} is surjective. Since $\dim \ker \Abb \cap X^{s+2,p} = 3$,
 the map \eqref{Ulem12} must, in fact, be an isomorphism.
\end{proof}

\subsect{compatexist}{Existence of initial data satisfying the compatibility conditions}

\begin{lem} \label{techA}
Suppose that $\Yc_0,\Yc_1,\ldots, \Yc_r$, and $\Zc$ are Banach spaces with continuous (linear) embeddings
\eqn{techA1}{
\iota_{i,j} : \Yc_{i} \longrightarrow \Yc_{j} \qquad  i,j\in \{\, 0,1,\ldots,r \, \},\; i<j,
}
$\Uc \subset \Yc_r$ is open, and $F\in C^{r+1}(\Uc,\Zc)$. Then the map defined by
\eqn{techA2}{
F_r(y_0,y_1,\ldots,y_r) := \frac{d^r}{dt^r}\Bigl|_{t=0} F(c(t)) \where c(t)= \sum_{j=0}^r t^j \iota_{j,r}(y_j)
}
is in $C^1(\iota_{0,r}^{-1}(\Uc)\times \prod_{j=1}^r \Yc_j,\Zc)$.
\end{lem}
\begin{proof}
Since $\Uc \in Y_r$ is open, it follow from the continuity of the map $\iota_{0,r}$ that
$\iota_{0,r}^{-1}(\Uc) \subset Y_0$ is open.
Next, fix $y_0 \in \iota_{0,r}^{-1}(\Uc)$ and $y_j \in Y_j$ for $j=1,2,\dots,r$. Then the continuity of
the maps $\iota_{j,r} : Y_j \rightarrow Y_r$ and $\iota_{0,r}(y_0)\in \Uc$ guarantees the existence of
a $\delta >0$ such that
\eqn{techA4}{
c(t) = \sum_{j=0}^r t^j \iota_{j,r}(y_j) \in \Uc \quad \forall \; t \in (-\delta,\delta).
}
Clearly, this implies that $c \in C^\infty((-\delta,\delta),\Uc)$, and hence, that the map
\eqn{techA6}{
\iota_{0,r}^{-1}(\Uc)\times \prod_{j=1}^r \Yc_j \ni (y_0,\ldots,y_r) \longmapsto \frac{d^r}{dt^r}\Bigl|_{t=0} F(c(t)) \in \Zc
}
is well defined and continuously differentiable.
\end{proof}

\begin{prop} \label{compatA}
Suppose $3<p<\infty$, $s\in \Zbb_{\geq 0}$, and $\Oct^{s+2,p}\subset W^{s+2,p}(\Bc,\Rbb^3)$ is the open neighborhood
of $\psi_0$ from theorem \ref{newtA}. Then the maps (see \eqref{funcs})
\alin{compatA1}{
&E \: :\: \Oct^{s+2,p} \longrightarrow W^{s,p}(\Bc,\Rbb^3), \\
&\EdB \: : \:  \Oct^{s+2,p} \longrightarrow B^{s+1-1/p,p}(\del \Bc)
}
are $C^\infty$.
\end{prop}
\begin{proof}
First, we recall that, by assumption $\taub^{iA}$, is a smooth function of its arguments $\del_A \phi^i$ in the neighborhood
of the identity map $\psi_0^i$. Since $p>3$ and $s\geq 0$, we have  that $s+1 > 3/p$, and it follows from theorem 1, Section 5.5.2, of \cite{RS}, and the continuity of differentiation (cf. \eqref{Sobdiff}) that the map $\Oct^{s+2,p}\ni \phi \longmapsto \tau^{iA} \in W^{s+1,p}(\Bc,\Rbb^6)$
is $C^\infty$. The proof then follows directly from the continuity of the trace map \eqref{trace}.
\end{proof}

Defining
\eqn{OcapX}{
\Oc^{s+2,p} = \Oct^{s+2,p} \cap X^{s+2,p},
}
we have that
\eqn{OcapX2}{
\psi_0 \in \Oc^{s+2,p}
}
by \eqref{Uclem3}. We also define\footnote{ For two Banach spaces $X$ and $Y$, $L(X,Y)$ denotes the set of continous linear maps from $X$ to $Y$.}
\eqn{Fdef1}{
F \: : \: \Oc^{s+2,p} \times  X^{s,p}\times  \Rbb \longrightarrow Y^{s,p}\times \Rbb^3 \: :\: (\phi_0,\phi_2,\ep) \longmapsto
\bigl(\Fc(\phi_0,\phi_2,\ep), B(\phi_0,\psi_0) \bigr)
}
where
\eqn{Fcdef}{
\Fc(\phi_0,\phi_2,\ep) =  \bigl(E(\phi_0)+\ep^2 \Lambda(\phi_0) - \phi_2, \EdB(\phi_0)\bigr),
}
and we let
\leqn{Pbbdef1}{
\Pbb : \Oc^{s+2,p} \rightarrow L(Y^{s,p},Y^{s,p})
}
denote any $C^\infty$ map for which $\Pbb_{\psi_0}$ coincides with the projection operator from theorem \ref{surjthm}. Furthermore,
we assume that for each $\phi_0 \in \Oc^{s+2,p}$, the linear operator
\leqn{Pbbdef2}{
\Pbb(\phi_0)|_{Y^{s,p}_{\phi_0}} \: : \: Y^{s,p}_{\phi_0} \longrightarrow Y^{s,p}_{\psi_0}
}
is an isomorphism and
\leqn{Pbbdef3}{
Y^{s,p} = Y^{s,p}_{\phi} \oplus \ker \Pbb(\phi_0).
}
The existence of a map \eqref{Pbbdef1} satisfying \eqref{Pbbdef2} and \eqref{Pbbdef3} can be found in \cite{LeDret}.

For $r\in \Zbb_{\geq 0}$, we define
\alin{phivec}{
&\phiv_r  = \bigl(\phi_0,\phi_1,\ldots,\phi_r\bigr)^T, \\
&\phi_r(t)  = \phi_0 + \sum_{j=1}^r \ep \frac{t^j}{j!}\phi_j, \\
&\phi_{r,2}(t)  = \sum_{j=0}^r \ep \frac{t^j}{j!} \phi_{j+2},
}
and
\alin{Fdefvec}{
&F_0(\phi_0,\phi_2,\ep)  = F(\phi_0,\phi_2,\ep), \\
&F_1(\phiv_1,\phi_3,\ep)  = \frac{1}{\ep}\left(\Fc_1(\phiv_1,\phi_{3},\ep),\phi_{1,2}(t),\ep),\ep B(\psi_0,\phi_1) \right), \\
&F_{r+2}(\phiv_{r},\phi_{r+2},\ep)  = \frac{1}{\ep} \left( \Fc_{r+2}(\phiv_{r+2},\phi_{r+4},\ep), \dt{r} B(\phi_0 (t),\phi_{0,2} (t)) \right) \quad (r\geq 0) \\
}
where
\eqn{Fcrdef}{
\Fc_r(\phiv_r,\phi_{r+2},\ep) =  \dt{r} \Fc(\phi_{r}(t),\phi_{r,2}(t),\ep).
}
\begin{rem} \label{compatrem}
Under the identification
\eqn{compatrem1}{
\phi_r = \bigl(\del_t^r\phi\bigr)\bigl|_{t=0},
}
it follows from the definition of $\phi_r(t)$ and $\phi_{r,2}(t)$ above that
\eqn{compatrem2}{
\dt{r}\phi_r(t) = \ep\bigl(\del_t^r\phi\bigr)\bigl|_{t=0} \AND \dt{r} \phi_{r,2}(t) = \ep\bigl(\del_t^{r+2}\phi\bigr)\bigl|_{t=0}.
}
Moreover, under this identification, we have that
\eqn{compatrem3}{
\Fc(\phi_0,\phi_2,\ep) = 0
}
if and only if
\eqn{compatrem4}{
\bigl(\del_A \taub^{iA}(\del\phi) - \ep^2 \delta^{ij}f_j^A\del_A\Ub(\phi) -\del_t^2 \phi^i\bigr)|_{t=0}   = 0 \AND
\bigl(\nu_A \taub^{iA}(\del\phi)|_{\del \Bc}\bigr)|_{t=0}  = 0.
}
Also, by repeatedly differentiating the equations of motion \eqref{evolve:1}-\eqref{evolve:3}, it is not
difficult to see that
\leqn{compatrem5}{
\Fc_r(\phiv_r,\phi_{r+2},\ep) = 0
}
if and only if
\eqn{compatrem6}{
\del_t^r \bigl(\del_A \tau^{iA}(\del\phi) - \ep^2 \delta^{ij}f_j^A\del_A\Ub(\phi) -\del_t^2 \phi^i\bigr)|_{t=0}   = 0 \AND
\del_t^r \bigl(\nu_A \taub^{iA}(\del\phi)|_{\del \Bc}\bigr)|_{t=0}  = 0.
}
This shows that solving \eqref{compatrem5} for $r=0,1,\ldots,\ell$ will produce initial data that satisfies the compatibility
conditions to order $\ell$.
\end{rem}

In order to use the implicit function theorem to solve the equations \eqref{compatrem5}, we need to introduce
the following maps which are a closely related to the $\Fc_r$ and $F_r$ maps introduced above:
\lalign{Gdefvec}{
& \Gc(\phi_0,\phi_2,\ep) = \Pbb(\phi_0)\Fc(\phi_0,\phi_2,\ep), \notag\\
& G_0(\phi_0,\phi_2,\ep) = (\Gc(\phi_0,\phi_2,\ep), B(\psi_0,\phi_0) ), \notag \\
&G_1(\phiv_1,\phi_3,\ep)  = \frac{1}{\ep}\left(\Gc_1(\phiv_1,\phi_{3},\ep),\ep B(\psi_0,\phi_1) \right), \notag \\
&G_{r+2}(\phiv_{r+2},\phi_{r+4},\ep) \nonumber \\
&\quad
 = \frac{1}{\ep} \left( \dt{r+2} \Gc_{r+2}(\phiv_{r+2}(t),\phi_{r+4},\ep), \dt{r} B(\phi_0 (t),\phi_{0,2} (t)) \right), \quad (r\geq 0) \notag
\intertext{and}
&\Gv_r(\phiv_r,\phi_{r+1},\phi_{r+2},\ep)  =  \begin{pmatrix}
G_0(\phiv_0,\phi_2,\ep) \\
G_{1}(\phiv_1,\phi_3,\ep) \\
\vdots \\
G_{r-1}(\phiv_{r-1},\phi_{r+1},\ep)\\
G_{r}(\phiv_r,\phi_{r+2},\ep)
\end{pmatrix}  \notag  
}
where
\eqn{Gcrdef1}{
\Gc_r(\phiv_r,\phi_{r+2},\ep) = \dt{r}\Gc(\phi_r(t),\phi_{r,2}(t),\ep).
}
\begin{prop} \label{propGF}
Suppose $s \in \Zbb_{\geq r}$ and $3<p<\infty$. Then the maps
\eqn{propGF1}{
F_r \: : \: \Oct^{s+2,p} \times \Bigl( \prod_{j=1}^{r+2} W^{s+2-j,p}(\Bc,\Rbb^3) \Bigr) \times \Rbb^3 \longmapsto Y^{s-r,p}\times \Rbb^3
}
and
\eqn{PropGF2}{
G_r \: : \: \Oct^{s+2,p} \times \Bigl( \prod_{j=1}^{r+2} W^{s+2-j,p}(\Bc,\Rbb^3) \Bigr) \times \Rbb^3 \longmapsto Y^{s-r,p}\times \Rbb^3
}
are $C^1$.
\end{prop}
\begin{proof}
First we note that the maps $F_0$ and $G_0$ are $C^\infty$ which follows from Corollary \ref{newtB}, proposition \ref{compatA},
and the smoothness of the map \eqref{Pbbdef1}. The proof then follows immediately from lemma \ref{techA} and the definition
of $F_r$ and $G_r$.
\end{proof}

Introducing
\eqn{Ecdef1}{
\Ec(\phi,\ep) = \bigl( E(\phi)+\ep^2\Lambda(\phi), \EdB(\phi)\bigr),
}
we get that
\eqn{Ecdef2}{
F_0(\phi_0,\phi_2,\ep) = \bigl(\Ec(\phi_0,\ep)-(\phi_2,0), B(\psi_0,\phi_0) \bigr),
}
and it follows easily from the definition of the $F_r$ maps above that
\eqn{Ecdef2a}{
F_1(\phiv_1,\phi_3,\ep) = \bigl(D_\phi\Ec(\phi_0,\ep)\cdot \phi_1 - (\phi_3,0),B(\psi_0,\phi_1) \bigr),
}
and
\eqn{Ecdef2b}{
F_2(\phiv_2,\phi_4,\ep) = \bigl(D_\phi\Ec(\phi_0,\ep)\cdot\phi_2 - (\phi_4,0) + \ep D^2_\phi\Ec(\phi_0,\ep)
\cdot(\phi_1,\phi_1), B(\phi_0,\phi_2) \bigr).
}
Proceeding inductively, we obtain
for $r \geq 1$,
\leqn{Ecdef3}{
F_r(\phiv_r,\phi_{r+2},\ep) = \bigl(D_\phi\Ec(\phi_0,\ep)\cdot \phi_r
-(\phi_{r+2},0), B(\phi_0,\phi_r) \bigr) + \ep \Ft_r(\phiv_{r+1},\ep)
}
where the map $\Ft_r$ is $C^1$.
Setting,
\eqn{Pbbdt}{
\Pbb_r(\phiv_r) = \dt{r} \Pbb(\phi_r(t)),
}
the product rule shows that
\leqn{Gcrdef2}{
\Gc_r(\phiv_r,\phi_{r+2},\ep) = \sum_{j=0}^r \binom{r}{s} \Pbb_{r-j}(\phiv_{r-j})\Fc_r(\phiv_j,\phi_{j+2},\ep).
}
This and \eqref{Ecdef3}, in turn, give
\leqn{Grdef}{
G_r(\phiv_r,\phi_{r+2},\ep) = \bigl(\Pbb(\phi_0)\bigl[D_\phi\Ec(\phi_0,\ep)\cdot \phi_r -(\phi_{r+2},0) \bigr], B(\phi_0,\phi_r) \bigr) + \ep \Gt_r(\phiv_{r+1},\ep) \quad r\geq 1
}
where $\Gt_r$ is $C^1$.

Letting
\eqn{psivdef}{
\psiv_r = (\psi_0,0,\ldots,0),
}
it follows directly from \eqref{Grdef} that the derivative of $\Gv_r$  evaluated at
$(\phiv_r,\phi_{r+1},\phi_{r+2},\ep) = (\psiv_r,0,0,0)$ is
\leqn{A1}{
D_{\phiv_r} \Gv_r(\psi_v,0,0,0) \cdot \delta \phiv_r = A \cdot \delta \phiv_r
}
where
\leqn{A2}{
A = \begin{pmatrix}
\Lc(\psi_0)               &   0          & \Mcal(\psi_0)   &      0        &  \cdots       & 0       \\
    0                     &  \Lc(\psi_0) &     0         & \Mcal(\psi_0)   &         &   \vdots \\
    0                     &       0      &   \Lc(\psi_0) &      0          & \ddots  &   0 \\
    \vdots                 &    \vdots          &       0       &   \Lc(\psi_0)   &         &  \Mcal(\psi_0)    \\
                           &              &               &                 &  \ddots &    0             \\
    0                      &      0       &      0       &    \cdots       &    0     & \Lc(\psi_0) \\
\end{pmatrix}
}
\lalign{A3}{
\Lc(\psi_0) \cdot \delta\psi &=\bigl( \Pbb(\psi_0) \Abb \delta \psi , B(\psi_0,\delta\psi) \bigr), \label{A3.1}
\intertext{and}
\Mcal(\psi_0) \cdot \delta\psi & = \bigl( -\Pbb(\psi_0)(\delta\psi,0), 0 \bigr). \notag
}
We note that in deriving \eqref{A1}-\eqref{A2}, we have used
\eqn{A4}{
\Ec(\psi_0,0) = 0\AND D_\phi \Ec(\psi_0,0)\cdot \delta\psi = \Abb\delta\psi.
}

Although our main objective is to solve the equations \eqref{compatrem5}, we
first solve $\Gv_r = 0$ and later show that this implies that compatibility conditions
are satisfied to order $r$. To solve
$\Gv_r=0$, we use the implicit function theorem. The proof we present is modeled on
the existence proof for static, self-gravitating elastic bodies presented in
\cite{BeSch03}.

\begin{prop} \label{impA}
There exists an $\ep_0>0$, open neighborhoods $\Nc_{r+1} \subset X^{s+1-r,p}$ and
$\Nc_{r+2} \subset X^{s-r,p}$ both containing $0$,  and $C^1$ maps
$$
\Phi_j \: :\: \Nc_{r+1}\times \Nc_{r+2} \times (-\ep_0,\ep_0)\longrightarrow X^{s+2-j,p}
\: :\: (\phi_{r+1},\phi_{r+2},\ep) \longmapsto
\Phi_j(\phi_{r+1},\phi_{r+2},\ep) ,
$$
for $j=0,1,\ldots,r$,
such that
\eqn{impA2}{
\Phi_0(0,0,0) = \psi_0, \quad \Phi_j(0,0,0) = 0, \quad j=1,2,\ldots,r
}
and
\eqn{impA3}{
\Gv_r\bigl(\Phiv_r(\phi_{r+1},\phi_{r+2}),\phi_{r+1},\phi_{r+2},\ep\bigr) = 0 \quad \forall \: (\phi_{r+1},\phi_{r+2},\ep) \in \Nc_{r+1}\times \Nc_{r+2} \times (-\ep_0,\ep_0)
}
where $\Phiv_r = (\Phi_0,\Phi_1,\ldots,\Phi_r)$.
\end{prop}
\begin{proof}
We begin by verifying the the linear operator \eqref{A2} is an isomorphism.
\begin{lem} \label{Alem}
For $s\in \Zbb_{\geq r}$ and $3<p<\infty$, the map
\eqn{Alem1}{
A \: : \: \prod_{j=0}^r X^{s+2-j,p} \longrightarrow \prod_{j=0}^r \bigl( Y^{s-j,p}\times \Rbb^3 \bigr)
}
is a linear isomorphism.
\end{lem}
\begin{proof}
By lemma \ref{Uclem}, $\psi_0 \in \Uc^{s,p}$ and so it follows from the definition of $\Lc(\psi_0)$ (see \eqref{A3.1})
and theorem \ref{surjthm}  that $\Lc(\psi_0) \: : \: X^{s+2-j,p} \longrightarrow Y^{s-j,p}\times \Rbb^3$
is an isomorphism for $j=0,1,\ldots r$. The proof now follows immediately from upper triangular structure
of the map $A$ (see \eqref{A2}).
\end{proof}
Recalling that $\Ec(\psi_0) = 0$, it is clear from the definition of $\Fc$ that
\eqn{impA4}{
\Fc(\psi_0,0,0) = 0,
}
and hence, by the antisymmetry of the map $B$, that
\eqn{impA5}{
G_0(\psi_0,0,0) = 0.
}
Furthermore, it clear from \eqref{Grdef} that
\eqn{impA6}{
G_j(\psiv_j,0,0) = 0 \quad j=1,2,\ldots,r,
}
which, in turn, shows that
\leqn{impA7}{
\Gv_r(\psiv_r,0,0,0) = 0.
}
The proof of the proposition now follows proposition \ref{propGF}, lemma \ref{Alem}, and the implicit function theorem.
\end{proof}

\begin{thm} \label{cthm}
The maps $\Phi_j$ $j=0,1,\ldots,r$ from proposition \ref{impA} satisfy
\gath{cthm1}{
\Fc_j\bigl(\Phiv_j(\phi_{r+1},\phi_{r+2}),\Phi_{j+2}(\phi_{r+1},\phi_{r+2}),\ep\bigr) = 0 \quad j=0,1,\ldots,r-2, \\
\Fc_{r-1}\bigl(\Phiv_{r-1}(\phi_{r+1},\phi_{r+2}),\phi_{r+1},\ep\bigr) = 0 \AND
\Fc_{r}\bigl(\Phiv_{r}(\phi_{r+1},\phi_{r+2}),\phi_{r+2},\ep\bigr) = 0
}
for all $(\phi_{r+1},\phi_{r+2},\ep) \in \Nc_{r+1}\times \Nc_{r+2} \times (-\ep_0,\ep_0)$.
\end{thm}
\begin{proof}
It is shown in \cite{BeSch03} that
\leqn{cthm2}{
\Cc_1\bigl(\Ec(\phi,\ep)\bigr) = 0 \AND
\Cc_2\bigl(\phi,\Ec(\phi,\ep)\bigr) = 0
}
are automatically satisfied
for all $\phi \in \Oc^{s+2,p}$ and $-\ep_0 < \ep < \ep_0$. Setting
\eqn{cthm3}{
\Ec_j(\phiv_j,\ep) = \dt{j} \Ec(\phi_r(t),\ep) \quad j=0,1,\ldots,r,
}
the formulas
\eqn{cthm4}{
\Cc_1\bigl(\Ec_j(\phiv_j,\ep) \bigr) = 0 \AND \sum_{k=0}^j \binom{j}{k} \Cc_2\bigl( \phi_k, \Ec_{j-k}(\phiv_{j-k},\ep)\bigr) = 0 \quad j=0,1,\ldots,r
}
follow from differentiating \eqref{cthm2}. In particular, this implies that
the maps  $\Phi_j$
from theorem \ref{impA} satisfy
\lgath{cthm5}{
\Cc_1\bigl(\Ec_j(\Phiv_j\bigl(\phi_{r+1},\phi_{r+2},\ep)),\ep\bigr) \bigr) =
0
\quad \text{ for } j=0,1,\ldots,r
\label{cthm5.1}
\intertext{and for $j=0,1,\ldots, r-2$,}
\sum_{k=0}^j \binom{j}{k} \Cc_2\bigl(
\Phi_k\bigl(\phi_{r+1},\phi_{r+2},\ep\bigr),
\Ec_{j-k}\bigl(\Phiv_{j-k}\bigr(\phi_{r+1},\phi_{r+2},\ep\bigl),\ep\bigr)\bigr)
= 0 ,
\label{cthm5.2}
}
for  all $(\phi_{r+1},\phi_{r+2},\ep) \in \Nc_{r+1}\times\Nc_{r+2}\times (-\ep_0,\ep_0)$

Next, we note that
\leqn{cthm6}{
\Cc_1((\phi_j,0)) = 0, \quad \forall \; \phi_r \in  X^{s+2-j,p},
}
and
\lalign{cthm7}{
\dt{j} B\bigl(\phi_{j}(t),\phi_{2,j}(t)\bigr) = 0 & \Longleftrightarrow \dt{j} \Cc_2\bigl(\phi_j(t),(\phi_{2,r}(t),0)\bigr)=0 \notag \\
& \Longleftrightarrow  \sum_{k=0}^j \binom{j}{k}\Cc_2\bigl(\phi_k,(\phi_{j+2-k},0)\bigr) = 0. \label{cthm7.1}
}
If we define
\eqn{cthm7a}{
\Phi_{r+1}(\phi_{r+1},\phi_{r+2},\ep) = \phi_{r+1} \AND  \Phi_{r+2}(\phi_{r+1},\phi_{r+2},\ep)  = \phi_{r+2},
}
then it follows from \eqref{cthm6}, \eqref{cthm7.1}, and the definition of $\Gv_r$ that the maps $\Phi_j$ satisfy
\lalign{cthm8}{
&\Cc_1\bigl(\Phi_j(\phi_{r+1},\phi_{r+2},\ep),0 \bigr) = 0 \quad j=0,1,\ldots,r+2, \label{cthm8.1}
\intertext{and}
&\sum_{k=0}^j \binom{j}{k} \Cc_2\bigl( \Phi_k(\phi_{r+1},\phi_{r+2},\ep), \bigl( \Phi_{j+2-k}(\phi_{r+1},\phi_{r+2},\ep), 0\bigr) \bigr) = 0
\quad j=0,1,\ldots,r \label{cthm8.2}
}
for all $(\phi_{r+1},\phi_{r+2},\ep) \in \Nc_{r+1}\times\Nc_{r+2}\times (-\ep_0,\ep_0)$.

Fixing $(\phi_{r+1},\phi_{r+2},\ep) \in \Nc_{r+1}\times\Nc_{r+2}\times (-\ep_0,\ep_0)$, the identities
 \eqref{cthm5.1}, \eqref{cthm5.2}, \eqref{cthm8.1}, \eqref{cthm8.2} and the definition of
the maps $\Fc_j$ show that
\lgath{cthm8a}{
\Cc_1\bigl(\Fc_j(\Phiv_j(\phi_{r+1},\phi_{r+2},\ep), \Phi_{j+2}(\phi_{r+1},\phi_{r+2},\ep),\ep \bigr) \bigr) = 0,
\label{cthm8a.1}
\intertext{and}
\sum_{k=0}^j \binom{j}{k} \Cc_2\bigl( \Phi_k(\phi_{r+1},\phi_{r+2},\ep),
\Fc_{j-k}(\Phiv_{j-k}(\phi_{r+1},\phi_{r+2},\ep), \Phi_{j+2-k}(\phi_{r+1},\phi_{r+2},\ep),\ep \bigr)\bigr) = 0
\label{cthm8a.2}
}
for $0\leq j \leq r$, while
\eqn{cthm9a}{
\Gc_0\bigl( \Phi_0(\phi_{r+1},\phi_{r+2},\ep) ,  \Phi_2(\phi_{r+1},\phi_{r+2},\ep), \ep \bigr) = 0
}
follows from theorem \ref{impA}, or equivalently, by the definition of $\Gc_0$, and
\leqn{cthm9}{
\Pbb\bigl(\Phi_0(\phi_{k+1},\phi_{r+2},\ep)\bigr) \Fc_0\bigl( \Phi_0(\phi_{r+1},\phi_{r+2},\ep) ,  \Phi_2(\phi_{r+1},\phi_{r+2},\ep), \ep \bigr) = 0.
}
But, we also have that
\eqn{cthm10z}{
\Fc_0\bigl( \Phi_0(\phi_{r+1},\phi_{r+2},\ep) ,  \Phi_2\bigl(\phi_{r+1},\phi_{r+2},\ep\bigr), \ep \bigr) \in Y^{s,p}_{\Phi_0(\phi_{k+1},\phi_{r+2},\ep)}
}
by \eqref{cthm8a.1}-\eqref{cthm8a.2}, and thus,
\eqn{cthm10a}{
\Fc_0\bigl( \Phi_0(\phi_{r+1},\phi_{r+2},\ep) ,  \Phi_2(\phi_{r+1},\phi_{r+2},\ep), \ep \bigr) = 0
}
by \eqref{Pbbdef3} and \eqref{cthm9}.

To finish the proof, we proceed by induction. So, we assume that
\leqn{cthm11}{
\Fc_0\bigl(\Phiv_j(\phi_{r+1},\phi_{r+2},\ep), \Phi_{j+2}(\phi_{r+1},\phi_{r+2},\ep), \ep \bigr) = 0  \quad
}
for $0\leq k \leq j < r$. Then
\leqn{cthm12}{
\Gc_k\bigl(\Phi_k(\phi_{r+1},\phi_{r+2},\ep),\Phi_{k+2}(\phi_{r+1},\phi_{r+2},\ep)\bigr)=0
}
for $0\leq k\leq j$ by theorem
\ref{impA}. Clearly,
\eqref{Gcrdef2}, \eqref{cthm11}, and \eqref{cthm12} imply that
\leqn{cthm13}{
\Pbb\bigl(\Phi_0(\phi_{k+1},\phi_{r+2},\ep)\bigr) \Fc_{j+1}\bigl( \Phiv_{j+1}(\phi_{r+1},\phi_{r+2},\ep) ,  \Phi_{j+2}(\phi_{r+1},\phi_{r+2},\ep), \ep \bigr) = 0.
}
But, since
\eqn{cthm10b}{
\Fc_{j+1}\bigl( \Phiv_{j+1}(\phi_{r+1},\phi_{r+2},\ep) ,  \Phi_{j+3}\bigl(\phi_{r+1},\phi_{r+2},\ep\bigr), \ep \bigr) \in Y^{s,p}_{\Phi_0(\phi_{k+1},\phi_{r+2},\ep)}
}
by \eqref{cthm8a.1}-\eqref{cthm8a.2}, we must, in fact, have
\eqn{cthm10}{
\Fc_{j+1}\bigl( \Phiv_{j+1}(\phi_{r+1},\phi_{r+2},\ep) ,  \Phi_{j+3}(\phi_{r+1},\phi_{r+2},\ep), \ep \bigr) = 0
}
by \eqref{Pbbdef3} and \eqref{cthm13}, and the proof is complete.
\end{proof}
In light of remark \ref{compatrem}, the following corollary is a direct consequence of the above theorem.
\begin{cor} \label{ccor}
Suppose $s \in \Zbb_{\geq 0}$, $3<p<\infty$, and the maps $\Phi_j$, the neighborhoods $\Nc_{s+1} \subset X^{1,p}\subset W^{1,p}(\Bc,\Rbb^3)$ and
$\Nc_{s+2} \subset X^{0,p}\subset W^{0,p}(\Bc,\Rbb^3)$, and $\ep_0$ are as in
proposition \ref{impA}. Then for
each $(\phi_{s+1},\phi_{s+2},\ep)\in \Nc_{s+1}\times\Nc_{s+2}\times (-\ep_0,\ep_0)$, the initial data
\eqn{corA}{
(\phi|_{t=0},\del_t|_{t=0}\phi) = \bigl(\Phi_0(\phi_{r+1},\phi_{r+2},\ep),\Phi_1(\phi_{r+1},\phi_{r+2},\ep)\bigr) \in \Oct^{s+2,p}\times \Oct^{s+1,p}
}
satisfy the compatibility conditions to order $s$.
\end{cor}

\sect{loc}{Local well-posedness}

In this section, we prove local well-posedness
for the system \eqref{evolve:1}-\eqref{evolve:3} using the approach of
\cite{koch}. The system under consideration is an initial-boundary value
problem of elliptic-hyperbolic type, due to the presence of the equation
\eqref{evolve:2} in the system, and hence the results of \cite{koch} do not apply
directly. However, the techniques of \cite{koch} are readily adapted to include non-local terms,
and we will present an outline of the proof of this fact below. To conclude this section, we will apply the resulting
existence theorem to establish the existence of dynamical solutions to \eqref{evolve:1}-\eqref{evolve:3}.

\subsection{Setup and notation}
\label{sec:setup}
For ease of reference, we adopt
the index and other notational conventions of
\cite{koch}, with some exceptions, as pointed out below. We are interested
only in the case $n=3$, but it is convenient to treat the case of general $n
\geq 3$. The number of components of the system of equations is in
the case of elasticity equal to $N = 3$, but the treatment below applies
to general $N$.

Let $0 \leq i,k \leq n$, $1 \leq \alpha,\beta \leq n$, $1 \leq j,l \leq
N$. We work in a coordinate system $(x_i)$ and let $t = x_0$. The summation
convention is used. Further, we shall denote the unknown field by $u$ rather
than $\phi$. Let $Du = (\partial_t u , \partial_{x_1} u , \dots,
\partial_{x_n} u)$, and $D_x u = (\partial_{x_1} u , \dots,
\partial_{x_n} u)$. Fix some $s > n/2 + 1$, $s$ an integer, and
consider the system in the domain $\Omega$ with
boundary $\Gamma$ of regularity class $C^{s+2}$, and
denote $\Omega_T = [0,T) \times \Omega$ and $\Gamma_T = [0,T) \times \Gamma$
where $0 < T \leq \infty$.

Following \cite{koch}, we consider equations of the form
\begin{subequations} \label{eq:koch:(1.1)}
\begin{align}
\partial_{x_i} F^i_j (t,x,u, Du) &= w_j [t,x,u,Du] & \text{ in } \Omega_T,\\
v_j F^i_j (t,x,u,Du) &= g_j(t,x,u,Du) & \text{ in } \Gamma_T,\\
u &= u_0, \quad \partial_t u = u_1 & \text{ in } \{0\} \times \Omega
\end{align}
\end{subequations}
where, in contrast to \cite{koch}, $w_j[t,x,u,Du]$ is a functional of $(u,Du)$
that we allow to be non-local. The
properties required of $w$ are specified below in Assumption $w$.

\subsubsection*{Assumption 1}
We assume $u_0 \in W^{s+1}(\Omega)$, $u_1 \in W^s(\Omega)$.
We assume $F, g \in C^{s+1}(U)$ where
$U$ is a neighborhood of the graph of $u_0, u_1, D_x u_0$ as in \cite[\S
  1]{koch}, and
$$
a^{ik}_{jl} = \frac{\partial F^i_j}{\partial ( \partial_{x_k} u^l )},
\quad
h^k_{jl} = \frac{\partial g_j}{\partial ( \partial_{x_k} u^l)}.
$$
We decompose $h^k_{jl} = h^{sk}_{jl} + h^{uk}_{jl}$ where $h^{sk}_{jl} =
h^{sk}_{(jl)}$ is the symmetric part and $h^{uk}_{jl} = h^{uk}_{[jl]}$ is the
  antisymmetric part.
We assume that the symmetric part $h^{sk}_{(jl)}$ is of the form
$h^{sk}_{(jl)} = \theta^k h_{jl}$ where $\theta(t,x,u,Du)$ is a vector field
which is tangential to $\Gamma_T$ and satisfies $\theta^0 = 1$, and
$h_{jl}$ is symmetric.

For the elastic body, we have
\begin{equation}\label{eq:elast-koch-h}
h^k_{jl} = 0
\end{equation}
while $a^{ik}_{jl}$ can be calculated in terms of the elasticity tensor
$L$, cf. \eqref{eq:elast-L}.

\subsubsection*{Assumptions 2-5}
The structural relations of the elastic body imply
the hyperbolicity of the system. In particular, for
the self-gravitating elastic body, the symmetry and coerciveness assumptions,
Assumptions 2 and 3 of \cite{koch}, hold. Further, Assumption 4 of \cite{koch} on the
time components $a^{00}_{jl}$ follow directly from the structure of the
elastic system. For a discussion of the compatibility conditions on the initial data, see Assumption 5
in \cite{koch}.

Following \cite[p. 25]{koch}, we introduce the spaces
$E^s$, $G^s$, and $F^s$ with norms
\begin{align*}
||u||_{E^s}(t) &= \left ( \sum_{i=0}^s ||\partial_t^i
u(t)||_{W^{s+1-i}}^2 \right )^{1/2}, \\
||u||_{G^s_{t_1,t_2}} &= \sup_{t_1 \leq t \leq t_2} || u||_{E^s}(t),
\intertext{and}
||u||_{F^s_{t_1,t_2}} &= \int_{t_1}^{t_2} ||u||_{E^s} (t) \, dt,
\end{align*}
respectively.

\subsubsection*{Assumption $w$}
For non-local $w$, we make the following further assumptions. We assume $w$ to be well defined if the graph
of $(u,Du)$ lies in a suitable subset of $U$, where $U$ as in Assumption 1
above. There, the following conditions are imposed.
\begin{enumerate}
\item
If
$$
u \in \cap_{0 \leq j \leq s} C^j([0,T],W^{s+1-j}(\Omega)),
$$
we have
$$
w[u,Du] \in \cap_{0 \leq j \leq s} C^j([0,T],W^{s-j}(\Omega)).
$$
\item We require the map $u \mapsto w[u,Du]$ to be Lipschitz in the above
topology.

In particular, see section \ref{sec:contract} below,
we shall make use of the estimate
\begin{equation}\label{eq:w-lip}
||\partial^2_t ( w(v_1) - w(v_0)) ||_{L^2} \leq c || (\partial_t ( v_1 - v_0) ||_{E^2}.
\end{equation}

\item
Finally, we require the following uniform estimate
$$
||w||_{E^s} \leq c ( 1 + ||D^2 u||_{L^\infty}) ( 1 + ||u||_{E^{s+1}} )
$$
where $c$ is a constant depending on $||Du||_{L^\infty}$
as well as the
coercivity constants $\kappa, \mu$ for the system.
\end{enumerate}

We have the following result, which is the analog of \cite[theorem 1.1]{koch}.
\begin{thm} \label{thm:koch:1.1}
$\;$

\begin{enumerate}
\item
Existence, regularity: There exists a unique $0 < t_0 \leq T$, and a unique
classical solution $u \in C^2 (\Omega_{t_0} \cup \Gamma_{t_0} )$ of
(\ref{eq:koch:(1.1)}) with $D^\sigma u(t) \in L^2(\Omega)$ if $0 \leq \sigma
\leq s+1$. Here $D^\sigma u$ denotes all derivatives of order $\sigma$.
\item Continuous dependence on initial data.
\item Blow up: $t_0$ is characterized by the two alternatives: either the
  graph of $(u, Du)$ is not precompact in $U$
or
$$
\int_0^t ||D^2 u(\tau)||_{L^\infty(\Omega)} d\tau \to \infty, \quad \text{ as
} t \to t_0.
$$
\end{enumerate}
\end{thm}

\subsection{Linear systems} \label{sec:linear}

For a solution to the system (analogous to \cite[(2.10)]{koch})
\begin{subequations}\label{eq:koch:(2.10)}
\begin{align}
\partial_{x_i} ( a^{ik} \partial_{x_k} u)  &=  w + \partial_{x_i}
f^i & \text{ in } \Omega_T, \\
v_\alpha (a^{\alpha k} \partial_{x_k} u - f^\alpha) &= h^{sk} \partial_{x_k} u & \text{ in } \Gamma_T
\end{align}
\end{subequations}
with the coercivity and structure conditions as in \cite[\S 2]{koch}, and in
particular, the regularity assumptions (see \cite[Assumption 2s, p. 26]{koch})
\begin{subequations}\label{eq:koch:Assumption:2s}
\begin{align}
a,h &\in G^s \cap C^1(\Omega_T) \quad \partial_t a, \partial_t h \in F^s, \\
w &\in G^{s-1} \cap W^{s,1} ([0,T], L^2(\Omega)) \cap
W^{s-1,1}([0,T],W^1(\Omega)) , \\
f^i &\in G^s, D f^i \in W^{s,1} ([0,T],L^2(\Omega)) , \\
u_0 &\in W^{s+1}(\Omega) \quad u_1 \in W^s(\Omega),
\end{align}
\end{subequations}
we have the estimate
\begin{multline}\label{eq:koch:(2.11)}
||u||_{E^{s+1}}(t_2) \leq \tilde{c} \bigg{(}
||u(t_1)||_{E^{s+1}} + ||w||_{G^{s-1}}
 \\
\left . + || f ||_{G^s} + \int_{t_1}^{t_2} ||\partial_t^s w(t) ||_{L^2(\Omega)} +
||\partial_t^{s+1} f^i (t) ||_{L^2(\Omega)} dt \right )
\end{multline}
where
$$
\tilde{c} = \tilde{c}(\kappa,\mu,||a||_{G^s\cap C^1(\Omega_T)},||\partial_t
a||_{F^s}, ||h^s||_{G^s\cap C^1(\Omega_T)}, ||\partial_t h^s||_{F^s}).
$$
Note that the system given in (\ref{eq:koch:(2.10)}) is of a
restricted form with $g = 0$, $h^u = 0$.
These terms can be absorbed into the
others, cf. the discussion in \cite[\S 2]{koch}. This is achieved by
introducing the modified coefficients $\bar{a}^{ik}, \bar{f}^i, \bar{w}$ as
in \cite[p. 31]{koch}, see also (\ref{eq:koch:p31}) below.

Applying the above estimate to a system of the form
\begin{subequations}\label{eq:koch:(2.10):mod}
\begin{align}
\partial_{x_i} ( a^{ik} \partial_{x_k} u)  &=  w + \partial_{x_i}
f^i & \text{ in } \Omega_T \\
v_\alpha (a^{\alpha k} \partial_{x_k} u - f^\alpha) &= h^{k} \partial_{x_k} u
+ g & \text{ in } \Gamma_T ,
\end{align}
\end{subequations}
we get, in view of the above discussion,
the inequality
\begin{multline}\label{eq:koch:(2.11):gen}
||u||_{E^{s+1}}(t_2) \leq \tilde{c} \bigg{(}
||u(t_1)||_{E^{s+1}} + ||w||_{G^{s-1}}
 + || f ||_{G^s} + ||g||_{G^s} \\
+
\int_{t_1}^{t_2} ||\partial_t^s w(t) ||_{L^2(\Omega)}
+ ||\partial_t^{s+1} f^i (t) ||_{L^2(\Omega)}
dt   \\
\left . + \int_{t_1}^{t_2} || \partial_t^s g||_{W^1(\Omega)}
+ || \partial_t^{s+1} g ||_{L^2(\Omega)} ds \right )
\end{multline}

\subsection{Proof of theorem \ref{thm:koch:1.1}} \label{sec:kochproof}
In this section, we discuss the main steps in the proof of theorem
\ref{thm:koch:1.1}. In the following calculations, we suppress the indices
$u_j$ on $u$ and the corresponding indices on $a^{ik}_{jl}, w_j$, etc.
First, we apply a time derivative to (\ref{eq:koch:(1.1)})
which gives\footnote{We write $\partial^2_{x_k, t}$ where $\partial_{x_k t}$
  is used in \cite{koch}}
\begin{subequations}\label{eq:koch:linearized}
\begin{align}
\partial_{x_i} ( a^{ik} \partial^2_{x_k,t} u)  &= \partial_t w + \partial_{x_i}
f^i & \text{ in } \Omega_T, \\
v_\alpha (a^{\alpha k} \partial^2_{x_k,t} u - f^\alpha) &= h^k \partial^2_{x_k,
  t} + \bar{g} & \text{ in } \Gamma_T ,
\end{align}
\end{subequations}
where
$$
f^i(t,x,u,Du) = a^{ik} \partial^2_{x_k,t} u - \partial_t F^i,
\quad \bar{g} = \partial_t g - h^k \partial^2_{x_k,t} u.
$$
For the elastic system, the coefficients have no explicit time dependence,
and the boundary condition is homogenous. So we have that
\begin{equation}\label{eq:elast-linearized}
f^i = 0, \quad \bar{g} = 0.
\end{equation}
Next, the system is rewritten in the form
\begin{subequations}\label{eq:koch:(3.1)}
\begin{align}
\partial_{x_i} ( \bar{a}^{ik} \partial^2_{x_k,t} u )&= \bar{w} + \partial_{x_i}
\bar{f}^i & \text{ in } \Omega_T, \\
v_\alpha (\bar{a}^{\alpha k} \partial^2_{x_k,t} u - \bar{f}^\alpha) &= h^{sk}
\partial^2_{x_k,t} u & \text{ in } \Gamma_T ,
\end{align}
\end{subequations}
where
\begin{subequations}\label{eq:koch:p31}
\begin{align}
\bar{f}^i &= f + v^i \bar{g} , \quad
\bar{a}^{ik} = a^{ik} - v^i h^{uk} + v^k h^{ui}
\\
\intertext{and}
\bar{w} &= \partial_t w - ((\partial_v + \div\, v) h^{uk} ) \partial^2_{x_k,t} u
+ (\partial_{x_i} v^\beta h^{ui} ) \partial^2_{x_\beta, t} u - \partial_v
\bar{g} - (\div\, v) \bar{g}. \label{eq:koch:p31:b}
\end{align}
\end{subequations}
\begin{rem}
By introducing the modifications $\bar{a}^{ik}, \bar{f}^i, \bar{w}$
of $a^{ik}, f^i, w$ as in (\ref{eq:koch:p31}), the resulting system
(\ref{eq:koch:(3.1)}) has no term $g$ and also $h^u = 0$. Thus it is
of the form of the system (\ref{eq:koch:(2.10)}) considered in
\cite[theorem 2.4]{koch}.
\end{rem}
For the elastic system, we have $f = g = h = 0$ and hence
\begin{equation}\label{eq:elast-(3.1)}
\bar{a}^{ik} =
a^{ik}, \quad \bar{w} = \partial_t w .
\end{equation}
We have that $\bar{a}, \bar{f}$ depend on $(t,x,u,Du)$ and $\bar{w}$ depends
linearly on $D^2 u$.

For technical reasons, we assume $s > n/2 +
2$ for the differentiability index $s$. This is one more degree of smoothness
than one would normally expect to require for a solution of a non-linear wave equation using energy estimates.
However, this assumption reflects the use of $v=\del_t u$ as the main variable in Koch's approach \cite{koch},
which has one less degree of differentiability compared to $u$.
The stated result for initial data with $s > n/2 + 1$ is recovered by
a smoothing and a limit argument, see the discussion in \cite[p. 33]{koch}.

Next, we let
$Y_{\tau,R}$ be the subset of
$$
H_\tau = \cap_{1 \leq i \leq s+1} W^{i,\infty} ([0,\tau],W^{s+1-i}(\Omega))
$$
of functions $v$ that satisfy $\partial_t^i v(0) = u_i$, where $u_i$ are the formal
time derivatives of $u$ for $0\leq i \leq s$ at $t=0$, and
$||v||_{H_\tau} \leq R$. By choosing $R$ sufficiently large, we can make sure
this set is non-empty.

The construction of solutions for the system (\ref{eq:koch:(1.1)}) makes use
of a standard fixed point argument, where one proves boundedness in a high norm and
contraction in a low norm. The high norm in this case is $W^{s+1}$, which for the linearized
(time-differen\-tiated) system corresponds to $W^s$.
The low norm for the time-differentiated system
is $W^2$.

This type of argument has been carried out for a
quasi-linear elliptic-hyperbolic system with no boundary conditions
in \cite{AML}. The difference between that system and the present situation
is that we have neumann-type boundary conditions, and the system has symbol
depending on $Du$, i.e. it is fully non-linear.

In the rest of this section, we consider the details of the contraction
estimate and the continuation property. The proof of the continuous
dependence given in \cite{koch} can be readily adapted to the present case
with the details given below.

\subsubsection{Contraction estimate} \label{sec:contract}
Define $J \in C(Y_{\tau,R}, H_\tau)$ as the map $v \mapsto u$, where $u$
solves the linear system\footnote{A typo in the boundary condition in \cite[eq. (3.2)]{koch}  is
  corrected here}
\begin{subequations} \label{eq:koch:(3.2)}
\begin{align}
\partial_{x_i} ( \bar{a}^{ik} (v) \partial^2_{x_k,t} u) &= \bar{w} (v) +
\partial_{x_i} \bar{f}^i & \text{ in } \Omega_T, \\
v_\alpha (\bar{a}^{\alpha k} (v) \partial^2_{x_k,t} u - \bar{f}^\alpha (v) )
= h^k(v) \partial^2_{x_k,t} u & \text{ in } \Gamma_T, \\
\partial_t u = u_1, \quad \partial_t^2 u = u_2, & \text{ in } \{0\} \times
\Omega
\end{align}
\end{subequations}
where we have denoted $a(v) = a(t,x,v,Dv)$ etc.

For the case of the elastic system, we again have
$$
\bar{a}^{ik} = a^{ik}, \quad \bar{w} = \partial_t w, \quad \bar{f} = 0, \quad
h = 0,
$$
and in particular,
$$
\bar{w} = D_u w . \partial_t u + D_{D_xu} w . D_x \partial_t u
$$
where we have used $D$ to denote the Frechet derivative. There is no explicit
$(t,x)$ dependence for this case.

Next, we observe that
$$
\partial_t w \in G^{s-1}, \quad \partial_t^{s+1} w \in F^0
$$
by Assumption $w$. From this, we see that for large $R$,
the image of $J$ lies in $Y$ provided $\tau$ is chosen small
enough.

We have
$$
Y \subset C^3([0,\tau],L^2 ) \cap C^2 ([0,\tau],W^1)
\cap C^1([0,\tau],W^2 )
$$
One shows that the set $Y$ is compact in the topology of the space defined by
the right hand side of this expression.

Let $u_0 = J(v_0)$, $u_1 = J(v_1)$. In order to derive the system solved by
$u_1 - u_0$, we let
$Z(u) = \partial_{x_i} \bar{a}^{ik} (t,x,u,Du) \partial^2_{x_k,t}
u$, and note that $Z(u_1) - Z(u_0)= \int_0^1 \frac{d}{d\lambda} Z(u_0 + \lambda (u_1 - u_0))$. Expanding this out gives
\begin{align*}
Z(u_1) - Z(u_0) &= \partial_{x_i} ( [\int_0^1 \bar{a}^{ik}_\lambda d\lambda]
\partial^2_{x_k,t} (u_1 - u_0) \\
&\quad + \partial_{x_i} ( \int_0^1 [\ta^{ik}_{\lambda,u} . (u_1 - u_0)] \partial^2_{x_k,t} (u_0
+ \lambda (u_1 - u_0)) )  d\lambda \\
&\quad + \partial_{x_i} ( \int_0^1 [\ta^{ik}_{\lambda, Du} . (Du_1 - Du_0)] \partial^2_{x_k,t} (u_0
+ \lambda (u_1 - u_0)) )  d\lambda
\end{align*}
where\footnote{A typo in
  \cite[p. 32]{koch} is corrected here.}
\begin{align*}
\bar{a}^{ik}_\lambda &=\bar{a}^{ik} (t,x,u_0 + \lambda(u_1 - u_0), D u_0 +
\lambda (D u_1 - D u_0)), \\
\bar{a}^{ik}_{\lambda,u} &= \frac{\partial \bar{a}}{\partial u}
(t,x,u_0 + \lambda(u_1 - u_0), D u_0 +
\lambda (D u_1 - D u_0)), \\
\bar{a}^{ik}_{\lambda,Du} &= \frac{\partial \bar{a}}{\partial Du}
(t,x,u_0 + \lambda(u_1 - u_0), D u_0 +
\lambda (D u_1 - D u_0)).
\end{align*}
Now, let
\begin{align*}
\ta^{ik} &= \int_0^1 \bar{a}^{ik}_\lambda d\lambda \\
\tw &= \bar{w}(v_1) - \bar{w}(v_0) \\
&\quad - \left [
\partial_{x_i} ( \int_0^1 [\bar{a}^{ik}_{\lambda,u} . (u_1 - u_0)] \partial^2_{x_k,t} (u_0
+ \lambda (u_1 - u_0)) )  d\lambda \right . \\
&\quad \left.  + \partial_{x_i} ( \int_0^1 [\bar{a}^{ik}_{\lambda, Du} . (Du_1 - Du_0)] \partial^2_{x_k,t} (u_0
+ \lambda (u_1 - u_0)) )  d\lambda \right ]
\\
\tf &= \bar{f}(v_1) - \bar{f}(v_0).
\end{align*}
Then $u_1 - u_0$ solves\footnote{This corrects a typo in \cite[p. 32]{koch}.}
\begin{subequations}\label{eq:koch:p32}
\begin{align}
\partial_{x_i} ( \ta^{ik} \partial^2_{x_k,t} (u_1 -u_0) ) &= \tw +
\partial_{x_i} \tf^i, \\
v_\alpha ( \ta^{\alpha k} \partial^2_{x_k,t} (u_1 - u_0) - \tf^{\alpha} ) &= \th^{k}
\partial^2_{x_k,t} (u_1 - u_0) + \tg
\end{align}
\end{subequations}
where $\th, \tg$ can be calculated along the same lines as above.

In particular, for the elastic system, we may calculate $\tw$ using $\bar{w}
= \partial_t w$, $\bar{h} = \bar{f} = \bar{g} = 0$. Note that $\tg$ is
non-vanishing in general due to contributions from $\bar{a}$.

We have the estimate\footnote{This corrects a typo in \cite[p. 32]{koch}., which
  leads to an incorrect estimate for $\partial_t (u_1 - u_0)$.}
\begin{equation}\label{eq:diffest}
|| \partial_t \tw(t) ||_{L^2} + ||\partial_t \tf ||_{E^1}(t) + || \partial_t
\tg ||_{E^1}(t) \leq c || \partial_t ( v_1 - v_0 ) ||_{E^2}
\end{equation}
where we made use of
$$
\tw = \bar{w}(v_1) - \bar{w}(v_0) + \text{terms involving $\bar{a}$}.
$$
As discussed above, $\bar{w}$ can be estimated given estimates for
$\partial_t w$, and hence $\partial_t \tw$ can be estimate in terms of
$\partial^2_t (w(v_1) - w(v_0)$. The terms involving $\bar{a}$ can be
estimated in terms of $u_1 - u_0$ and hence can be absorbed when applying
Gronwall.

From Assumption $w$, (\ref{eq:w-lip}), we have
$$
||\partial^2_t ( w(v_1) - w(v_0)) ||_{L^2} \leq c || \partial_t ( v_1 - v_0) ||_{E^2}.
$$
The required contraction estimate is obtained by applying the estimate for the
linear system given in \cite[theorem
  2.4]{koch}, as discussed in section \ref{sec:linear}, to the system
(\ref{eq:koch:p32}).
One checks that after suitable modifications, cf. section \ref{sec:linear},
the assumptions of \cite[theorem 2.4]{koch} holds for this
system, and this provides the needed contraction estimates.

Applying the inequality (\ref{eq:koch:(2.11):gen}) with $s=1$
to the system (\ref{eq:koch:p32}), we get the inequality
\begin{align*}
||\partial_t (u_1 - u_0) ||_{E^2} &\leq c \bigg{(} ||u_1(0) - u_0(0)||_{E^2}
+ ||\tw||_{G^0} + ||\tf||_{G^1} + ||\tg||_{G^1} \\
&\quad +
\int_0^t || \partial_t \tw (\sigma)
||_{L^2} + || \partial_t \tf^i ||_{E^1} + || \partial_t \tg (\sigma)||_{E^1}
d\sigma \bigg{)}
\\
\intertext{use that $u_1, u_0$ have the same initial data, and that
$\tw, \tf, \tg$ vanish at $t=0$}
&\leq c \left ( \int_0^t || \partial_t \tw (\sigma)
||_{L^2} + || \partial_t \tf^i ||_{E^1} + || \partial_t \tg (\sigma)||_{E^1}
d\sigma \right )  \\
\intertext{use (\ref{eq:diffest})}
&\leq c ||\partial_t (v_1 - v_0)||_{E^2}
\end{align*}
where we made use of the fact that $\tw,\tf,\tg$ all vanish at $t=0$.

\subsubsection{Continuation principle}
We next consider the proof of the continuation principle. Suppose the graph
of $(u, Du)$ lies in a compact set $U$. We need an estimate
of the following form, cf. \cite[eq. (3.3)]{koch}
\begin{equation}\label{eq:koch:(3.3)}
||u(t)||_{E^{s+1}} \leq \tilde{c} \left ( U, \int_0^t ||D^2 u(\tau)||_{L^\infty}
d\tau \right ) ( 1 + ||u_0||_{W^{s+1}} + ||u_1||_{W^s} ).
\end{equation}
This estimate is proved by applying operators $D^s_P$ of order $s$, tangent
to $\Gamma_\tau$, to both sides of the equation (\ref{eq:koch:(1.1)}),
and applying the estimate for
the linear system.
Let $s > n/2 + 1$, and let $u \in G^{s+1}$ be a solution of the system
(\ref{eq:koch:(1.1)}). We then have $u \in \cap_{0 \leq j \leq s+1}
C^j([0,T],W^{s+1-j} )$.
One finds, cf. \cite[p. 34]{koch},
that $D^s_P u$ solves an equation of the form
\begin{subequations}\label{eq:koch:p34}
\begin{align}
\partial_{x_i} (\bar{a}^{ik} \partial_{x^k} D^s_P u) &= \hw + \partial_{x_i}
\hf^i & \text{ in } \Omega_T, \\
v_i (\bar{a}^{ik} \partial_{x_k} D^s_P u - \hf^i) &= h^{sk} \partial_{x_k}
D^s_P u & \text{ in } \Gamma_T, \\
D^s_P u(0) & = u_1^s, \quad \partial_t D^s_P u = u_2^s
\end{align}
\end{subequations}
where $u_1^s, u_2^s$ are obtained by formal calculations from the initial
data $u_0, u_1$.
Here a transformation which absorbs the terms $g$ and $h^{uk}$ has been
applied, along the lines discussed in section \ref{sec:linear}, see
\cite[p.34]{koch} for details.

The basic energy estimate for systems of this type, see \cite[theorem
  2.2]{koch}, gives an estimate of the form\footnote{The norms $||\cdot||_2$
  in \cite[p. 35]{koch} should be $||\cdot ||_{L^2}$.}
\begin{equation}\label{eq:koch:p35}
||u||_{E^{s+1}} \leq c \left ( || u ||_{E^{s+1}}(0) + \int_0^t
||\hw||_{L^2}(\tau)
+ ||\partial_t \hf^i||_{L^2}(\tau)  d\tau \right )
\end{equation}
As shown in \cite{koch}, the $L^2$ norms in the right hand side of
(\ref{eq:koch:p35}) that involve local expressions can be estimated at a fixed
time in terms
of
$$
(1+ ||D^2 u||_{L^\infty} ) ( 1 + ||u(t)||_{E^{s+1}} )
$$
with a constant depending on $||Du||_{L^\infty}$ as well as the coercivity
constants $\kappa,\mu$. The term $\hw$ contains $D^s_P w$, and thus we
need the nonlocal term $w$ to satisfy at a fixed time
an estimate of precisely this form,
namely
$$
||D^s_P w||_{L^2} \leq c ( 1 + ||D^2 u||_{L^\infty}) ( 1 + ||u||_{E^{s+1}})
$$
which we can state as
$$
||w||_{E^s} \leq c ( 1 + ||D^2 u||_{L^\infty}) ( 1 + ||u||_{E^{s+1}} ).
$$
This estimate holds by Assumption $w$.

\subsection{Application to the elastic system}\label{locexist}


We are now ready to apply the local existence theorem \ref{thm:koch:1.1} to our system \eqref{evolve:1}-\eqref{evolve:3}.
\begin{thm} \label{lethm}
Suppose $s\in \Zbb_{\geq 3}$ $1<p<\infty$, and let $(\phi^i|_{t=0},\del_t|_{t=0}\phi^i)=(\phi^i_0,\phi^i_1)\in \Oct^{s+2,p}\times \Oct^{s+1,p} \subset W^{s+2}(\Bc)\times W^{s+1}(\Bc)$ be the
initial data from Corollary \ref{ccor}. Then there exists a $t_0>0$ and a unique classical solution $\phi^i\in C^2(\Bc_{t_0}\cup \del\Bc_{t_0})$
of \eqref{evolve:1}-\eqref{evolve:3} with $D^\sigma\phi^i(t) \in L^2(\Bc)$ for all $0\leq |\sigma| \leq s+1$ and $0\leq t < t_0$.
\end{thm}
\begin{proof} First, we observe by Corollary \ref{newtB} that non-local function $\Lambda^i(\phi) = -\delta^{ij}f^A_i\del_A \Ub(\phi)$ satisfies \emph{Assumption w}
of section \ref{sec:setup} for $\phi^i$ in the open set $\Oct^{s+1} \subset W^{s+1}(\Bc)$. Since the initial data
\eqn{lethem.1}{
(\phi^i|_{t=0},\del_t|_{t=0}\phi^i)=(\phi^i_0,\phi^i_1)\in \Oct^{s+2,p}\times \Oct^{s+1,p} \subset \Oct^{s+2,2}\times \Oct^{s+1,2} \subset W^{s+2}(\Bc)\times W^{s+1}(\Bc)
}
 from Corollary \ref{ccor} satisfies the compatibility conditions to order $s$, the proof then follows directly from
 theorem \ref{thm:koch:1.1}.
\end{proof}

\subsection*{Acknowledgements}
Part of this work was completed during visits of the authors T.A.O. and B.G.S.
to the Albert Einstein Institute. We are grateful to the Institute
its support and hospitality during these visits.

\appendix

\sect{sec:prelim}{Function Spaces}


\subsect{Sobspace}{$W^{k,p}$ spaces}
Give a finite dimensional vector space $V$, an open subset $\Omega \subset \Rbb^3$
with a $C^\infty$ boundary, $k\in \Zbb$, and $1\leq p\leq \infty$, we let $W^{k,p}(\Omega,V)$ denote the
standard Sobolev space for maps $u : \Omega \rightarrow V$. If $V= \Rbb$, then we will just write $W^{k,p}(\Bc)$, while
if $p=2$ we set $W^k(\Omega,V)=W^{k,2}(\Omega,V)$.

For these spaces, we recall the following  results:
\begin{itemize}
\item[(i)] Differentiation
\leqn{Sobdiff}{
\del_A = \frac{\del\;\;}{\del X^A} \; : \; W^{k,p}(\Omega) \longrightarrow W^{k-1,p}(\Omega)  \qquad A=1,2,3
}
defines a continuous linear map.
\item[(ii)] (Multiplication Inequality) If $1\leq p < \infty$, $k_1+k_2>0$, $k_1, k_2\geq k_3$, $k_3 < k_1+k_2 - 3/p$,
then there exists a $C>0$ such that
\leqn{mlem}{
\norm{uv}_{W^{k_3,p}(\Omega)} \leq C \norm{u}_{W^{k_1,p}(\Omega)}\norm{v}_{W^{k_2,p}(\Omega)}
}
for all $u\in W^{k_1,p}(\Omega)$ and $v \in W^{k_2,p}(\Omega)$.
\item[(iii)] Letting $B^{k,p}(\del \Omega)$ to denote the Besov spaces on the boundary $\del \Omega$, the
trace map
\leqn{trace}{
W^{k,p}(\Omega) \ni u \longmapsto \tr_{\del\Omega} u \in B^{k-1/p,p}(\del \Omega)
}
is well defined and continuous for $1<p<\infty$, and $k-1/p>0$ (cf. theorem 7.70 in \cite{Adams}).
\end{itemize}

\subsect{WSobspace}{$W^{k,p}_\delta$ spaces}
For $k\in \Zbb_{\geq 0}$, $1\leq p \leq \infty$, and $\delta \in \Rbb$, we use $W^{k,p}_\delta(\Rbb^3,V)$ to denote
the weighted Sobolev spaces for maps $u: \Rbb^3 \rightarrow V$ as defined in \cite{Bart86}. For open sets $\Omega \subset \Rbb^3$
for which $\Rbb^3\setminus \Omega$ is bounded, we denote the restriction of the weighted spaces to these
subsets by $W^{k,p}_\delta(\Omega,V)$.
Following \cite{Bart86}, the negative index spaces $W^{-k,p}_\delta(\Rbb^3)$ for $k\in \Zbb_{>0}$, $1<p<\infty$,
and $\delta \in \Rbb$ are defined
by duality.

We recall the following facts about the weighted Sobolev spaces:
\begin{itemize}
\item[(i)] Differentiation
\leqn{wSobdiff}{
\del_A  : W^{k,p}_\delta(\Omega) \longrightarrow W^{k-1,p}_{\delta-1}(\Omega)
}
defines a continuous linear map.
\item[(ii)] (Sobolev Inequality) If $k>3/p$, then there exists a $C>0$ such that
\leqn{wSob}{
\norm{u}_{L^\infty_\delta(\Omega,V)} \leq C\norm{u}_{W^{k,p}_\delta(\Omega,V)}
}
for all $u\in W^{k,p}_\delta(\Rbb^3,V)$. Moreover, $u\in C^0_\delta(\Rbb^3,V)$ and  $u(X)=\text{o}(|X|^\delta)$ as\footnote{Here $|X|=\sqrt{\delta_{AB}X^A X^B}$.} $|X|\rightarrow \infty$.
(Se lemma A.7 in \cite{Oli06}.)
\item[(iii)] (Multiplication Inequality) If $1\leq p < \infty$,  $k_1+k_2>0$, $k_1,k_2\geq k_3$, $k_3 < k_1+k_2 - 3/p$, and $\delta_1+\delta_2 \leq \delta_3$,
then there exists a constant $C>0$ such that \leqn{wmlem}{
\norm{uv}_{W^{k_3,p}_{\delta_3}(\Omega)} \leq C \norm{u}_{W^{k_1,p}_{\delta_1}(\Omega)}\norm{v}_{W^{k_2,p}_{\delta_2}(\Omega)}
}
for all $u\in W^{k_1,p}_{\delta_1}(\Omega)$ and $v \in W^{k_2,p}_{\delta_2}(\Omega)$. (See lemma A.8 in \cite{Oli06}.)

\item[(iv)] For $1<p<\infty$, $-1<\delta<0$, and $k\in \Rbb$, the Laplacian
\leqn{Deltaiso}{
\Delta = \delta^{AB}\del_A \del_B \: : \:  W^{k+2,p}_{\delta}(\Rbb^3) \longrightarrow W^{k,p}_{\delta-2}(\Rbb^3)
}
is a linear isomorphism with inverse given by the formula
\leqn{invLap}{
[\Delta^{-1}(u)](X) = -\frac{1}{4\pi}\int_{\Rbb^3} \frac{u(Y)}{|X-Y|} d^3 Y.
}
(The proof of this statement follows from using the fact that $\Delta \: : \:  W^{k+2,p}_{\delta}(\Rbb^3) \rightarrow W^{k,p}_{\delta-2}(\Rbb^3)$
is an isomorphism for $k\in \Zbb_{\geq 0}$ and $-1<\delta<0$ (cf. \cite[proposition 2.2]{Bart86}) together with duality and interpolation\footnote{As noted by Maxwell \cite{Maxwell}, the weighted spaces $W^{k,p}_\delta(\Rbb^3)$ and their fractional extensions
correspond to the spaces $h^k_{p,p(k-\delta)-3}$ in \cite{TriebelI,TriebelII} (cf. remark 2 and theorem 2 in \cite{TriebelII}). The
following  duality and interpolation results follow from remark 2, and theorems 2 and 3 in \cite{TriebelII}:
\begin{itemize}
\item[(a)] For $1<p<\infty$ and $q=p/(p-1)$, $W^{-k,q}_{-3-\delta}(\Rbb^3)$ is the dual of $W^{k,p}_\delta(\Rbb^3)$.
\item[(b)] If $1<p_1<\infty$, $1<p_2<\infty$ $0<\theta<1$, $k=(1-\theta)k_1 + \theta k_2$, $\delta = (1-\theta)\delta_1 + \theta\delta_2$, and $1/p=(1-\theta)/p_1 + \theta/ P_2$,
then $W^{k,p}_\delta(\Rbb^3)$ is the interpolation space $[W^{k_1,p}_{\delta_1}(\Rbb^3),W^{k_2,p}_{\delta_2}(\Rbb^3)]_\theta$.
\end{itemize}
}.)
\end{itemize}
For  bounded $\Omega\subset \Rbb^3$, we let
$\Ebb_{\Omega}$ denote an extension operator that satisfies
\leqn{ext2}{
\norm{\Ebb_{\Omega}(u)}_{W^{k,p}_{-10}(\Rbb^3,V)} \leq C_{k,p} \norm{u}_{W^{k,p}(\Omega,V)} \quad \forall \; u \in W^{k,p}(\Omega,V),
}
and
\leqn{ext1}{
\bigl[\del_A^\alpha \Ebb_{\Omega}(u)\bigr]|_{\Omega} = \del_A^\alpha u \quad \forall \; u \in  W^{k,p}(\Omega,V),\: |\alpha|\leq k.
}

\bibliographystyle{amsplain}
\bibliography{dyn6}

\end{document}